\documentclass[12pt]{article}
\usepackage[doublespacing]{setspace}
\usepackage{graphicx}
\usepackage{epsfig,longtable}
\usepackage{makeidx}
\usepackage{varioref}
\usepackage{xr}
\usepackage{amsmath}
\usepackage{amsthm}
\usepackage{amssymb}
\usepackage{natbib}
\usepackage{setspace}
\usepackage{subfig}
\usepackage{multirow}
\newtheorem{definition}{Definition}[section]
\newtheorem{procedure}{Procedure}[section]
\newtheorem{theorem}{Theorem}[section]

\newtheorem{lemma}{Lemma}[section]
\newtheorem{remark}{Remark}[section]
\newtheorem{ex}{Example}[section]

\numberwithin{equation}{section}

\input{tcilatex}

\begin{document}

{\sffamily\bfseries\Large Discovering findings that replicate from a
primary study of high dimension to a follow-up study}

\noindent%
\textsf{Marina Bogomolov and Ruth Heller}%
\footnote{\textit{Address for correspondence:} Department of
Statistics and Operations Research, Tel-Aviv university, Tel-Aviv,
Israel.  \ \textsf{E-mail:} ruheller@post.tau.ac.il. \ This work was supported by grant no. 2012896 from the
Israel Science
Foundation (ISF).  The authors thank Yoav Benjamini, Daniel Yekutieli,  and the referees for helpful comments.}\\

\noindent%

\textsf{Technion and Tel-Aviv University}

\thispagestyle{empty}%

\noindent%

\textsf{Abstract. \ We consider the problem of
 identifying whether
findings replicate from one study of high dimension to another, when
the primary study guides the selection of hypotheses to be examined
in the follow-up study as well as when there is no division of roles
into the primary and the follow-up study. We show that existing
meta-analysis methods are not appropriate for this problem, and
suggest novel methods instead. We prove that our multiple testing
procedures control for appropriate error-rates.
The suggested FWER controlling procedure is valid for arbitrary dependence among the test statistics within each study. A more powerful procedure is suggested for FDR control. We prove that this procedure controls the FDR if the test statistics are independent within the primary study, and independent or have dependence of type PRDS in the follow-up study.
For arbitrary dependence within the primary study, and  either arbitrary dependence or dependence of type PRDS in the follow-up study, simple conservative modifications of the procedure control the FDR.  We demonstrate the usefulness of these
procedures via simulations and real data examples. \
\medskip }

\noindent

\textsf{Keywords: \ False discovery rate; genome-wide association
studies ; meta-analysis; multiple comparisons;  replicability
analysis }\newpage

\setcounter{page}{1}

 \section{Introduction}
In genomics research, it is customary that a primary study is
followed by an independent
 study. Reporting results from the primary study, and then
 reporting the evidence from the follow-up study that supports these results,
gives a sense of the replicability of the results. For example,
findings are informally regarded as replicated if the $p$-value for
testing a null hypothesis is small in the primary study, and then
for the same hypothesis the $p$-value is fairly small in the
follow-up study. 

Many approaches are available for analyzing two or more studies,
where the follow-up studies simply serve to add power. See
\cite{Hedges85}, \cite{BY05}, \cite{skol06}, and \cite{Cardon07},
among others. In this work, we focus on analyzing two studies, where
the follow-up study serves to confirm the findings that were
identified in the primary study. A formal statistical
 approach is proposed for evaluating
 whether results from a primary study were indeed replicated in a
 follow-up study.

In observational studies, an association may fail to replicate
because the discovered association was not the actual effect of a
treatment but rather that of bias \citep{rosenbaum01}. However, if
the finding is replicated in a different cohort, using different
diagnostic or laboratory methods, then the association between
effect and outcome may be more convincingly causal.
\cite{rosenbaum01} gives the example of radiation and leukemia.
Suppose higher rates of leukemia are discovered in a primary study
among radiologists, and in a follow-up study among survivors at
Hiroshima and Nagasaki. Radiation is more convincingly causal if the
association discovered was replicated in the follow-up study, since
if radiation was not a cause of leukemia, then higher rates of
leukemia among radiologists would not lead us to expect higher rates
of leukemia among survivors at Hiroshima and Nagasaki. Another
example comes from the field of genomic research. Genome-wide
association studies (GWAS)  are observational studies, and therefore
there is always a danger that bias may explain away the discoveries.
\cite{kraft09} note that for common variants, the anticipated
effects are modest and very similar in magnitude to the subtle
biases that may affect genetic association studies - most notably
population stratification bias. For this reason, they argue that it
is important to see the association in other studies conducted using
a similar, but not identical, study base.

It is common practice that interesting findings in a primary GWA
study are investigated in another study, and the interesting
results of both studies are reported \citep{lander95}. 
For example, to discover  association between single-nucleotide
polymorphisms (SNPs) and hippocampal volume, \cite{Bis11} tested
$2.5\times 10^6$ SNPs in a primary study, and only a handful of SNPs
in promising loci in a follow-up study. \cite{Bis11} forwarded a SNP
for replication if the SNP $p$-value in the primary study was below
$4\times 10^{-7}$, corresponding to one expected false positive if
all SNPs are not associated with hippocampal volume. They viewed the
SNP as containing evidence of replication if its $p$-value in the
follow-up study was below 0.01, which is the Bonferroni threshold
when 5 hypotheses are simultaneously tested at the 0.05 family-wise
error rate (FWER). Their approach selects hypotheses for follow-up
based on suggestive evidence \citep{lander95}, and corrects for
multiplicity only in the follow-up study when discussing evidence of
replicability. Another naive approach is the following: apply a
multiple testing procedure within each study separately, and declare
as replicated the common findings. This approach will lead to
declaring SNPs that were found to be associated with the disease in
the primary study as well as in the follow-up study as the
discoveries of interest. If there was no danger that a multiple
testing procedure produces false positives, then this naive approach
would have been appropriate. However, multiple testing procedures
have a non-zero probability of producing false positives, unless
they have no power. Therefore, an approach that provides control
over false positives in each study separately, does not guarantee
control over false positives for evaluating whether the results were
replicated. Figure \ref{fig050508}, left panel, shows that the FDR
level can be as high as one when naively declaring results as
replicated if they were discovered by applying an FDR controlling
procedure at the nominal 0.05 level separately in each study.
Moreover, reducing the nominal 0.05 level does not resolve the
problem, see Remark \ref{rem-Naive}.

The paper is organized as follows. Section \ref{sec-notation} gives
the notation and review. Section \ref{sec-oneavailable} suggests
novel multiple testing procedures for replicability analysis, when
the primary study guides the selection of hypotheses to be examined
in a follow-up study. Section \ref{sec-bothavailableFDR} considers
the setting where there is no division of roles into a primary and a
follow-up study. In Section \ref{sec-example}, we revisit the
example of \cite{Bis11}. We also analyze an additional GWAS study, and show
additional examples from the  GWAS simulator HAPGEN2 \citep{Su11}.
Section \ref{sec-sim} describes a simulation study, and Section
\ref{sec-discussion} gives some final remarks.

\section{ Notation, Goal, and Review}\label{sec-notation}
 Consider a family of $m$ elementary null hypotheses
$H_{1},\ldots,H_{m}$. These elementary null hypotheses, or a subset
thereof, are tested in each of two independent studies. Let $h_{ij}$
be the indicator of whether $H_j$ is false in study $i$. The pair of
indicators $(h_{1j}, h_{2j})$ identifies  four  possible settings for
each $j$,
\begin{equation*} (h_{1j}, h_{2j}) = \left\{
\begin{array}{rl}
(0,0) & \text{if $H_j$ is true in both studies},\\
(1,0) & \text{if $H_j$ is false in the primary study but true in the follow-up study},\\
(0,1) & \text{if $H_j$ is true in the primary study but false in the follow-up study},\\
(1,1) & \text{if $H_j$ is false in both studies}.\\
\end{array} \right.
\end{equation*}
The set of indices $\{1,\ldots,m\}$ of the elementary null
hypotheses may be divided into four (unknown) subsets $I_{00}\cup
I_{10} \cup I_{01} \cup I_{11} = \{1,\ldots, m \}$, where each index
$j$ is in exactly one of the four subsets, defined as follows:
$I_{00} = \{j:  (h_{1j}, h_{2j}) = (0,0), j\in \{1,\ldots,m \} \};
I_{10} = \{j:  (h_{1j}, h_{2j}) = (1,0), j\in \{1,\ldots,m \} \};
I_{01} = \{j:  (h_{1j}, h_{2j}) = (0,1), j\in \{1,\ldots,m \} \};
I_{11} = \{j:  (h_{1j}, h_{2j}) = (1,1), j\in \{1,\ldots,m \} \}.$

\begin{definition} The {\em no replicability null hypothesis} for
elementary hypothesis $H_j$ is  $$ H_{NR,j}: (h_{1j}, h_{2j})\in
\{(0,0), (0,1), (1,0) \}. $$
\end{definition}
By definition, $H_{NR,j}$ is false if and only if the elementary
null hypothesis $H_j$ is false in both studies considered. In the
family of $m$ composite null hypotheses $H_{NR,1},\ldots,H_{NR,m}$,
the sets of indices of true and false  null hypotheses are
$I_{00}\cup I_{01} \cup I_{10}$ and $I_{11}$ respectively. Our goal
is to discover as many indices from $I_{11}$ as possible, i.e. true
positives, while controlling for the number of discoveries from
$I_{00}\cup I_{01} \cup I_{10}$, i.e. false positives.

Let $p_{ij}$ be the  $p$-value for the $j$th SNP in study $i$, for $i=1,2$.  Since the studies are
independent, the $p$-values are independent across studies. However,
the $p$-values
within each study may be dependent.
Inequality $x\geq y$ for vectors $x$ and $y$ is understood
componentwise.
\begin{remark}
In a typical meta-analysis \citep{Hedges85}, the goal is to discover
as many indices from $ I_{01} \cup I_{10}\cup I_{11}$ as possible, while
controlling for the number of discoveries from $I_{00}$. Had we
known, and had it been true,  that $I_{01} = \emptyset $ and $I_{10}
=\emptyset$, then the typical methods for meta-analysis could serve
to discover replicable findings. However, it is not known in
practice whether $I_{01}$ and $I_{10}$ are empty sets, and they need
not be empty  when the follow-up study is different, in at least one
aspect of design, from the primary study. Therefore, typical
meta-analysis methods are not appropriate when the aim is to
discover hypotheses with indices in $I_{11}$, treating all
discoveries from $I_{01}$ and $I_{10}$, in addition to $I_{00}$, as
false discoveries.
\end{remark}
\subsection{The partial conjunction
approach}\label{subsec-partialconj}
 In \cite{benjamini09} the partial conjunction approach \citep{conj} has been
 suggested for replicability analysis when $n\geq 2$ studies are
 available that examine the same problem. When exactly two studies are available, the procedure in \cite{benjamini09} amounts to
 applying the Benjamini-Hochberg false discovery rate (FDR) controlling procedure
\citep{yoav1}, henceforth referred to as the BH procedure,
 on the maximum of the two study $p$-values.
 However, this procedure may be too conservative, making it practically very
 difficult to discover false no replicability null hypotheses.

As an example, suppose there is an original
 GWA study that examines the association of $10^6$ SNPs with a phenotype. Now suppose 200 promising SNPs were
 selected to be examined in a follow-up study. If
 a SNP has a $p$-value of $0.025/10^6$ in the first study, and
 of $0.025/200$ in the second study, then the maximum $p$-value is
 $0.025/200$. The BH procedure will, most probably, not reject the no replicability null hypothesis for a SNP with maximum
 $p$-value of
 $0.025/200$, since this maximum $p$-value is not strong enough evidence when faced with
 $10^6$ hypotheses, out of which most of the hypotheses are true no replicability null hypotheses. The alternative procedures we suggest in Sections \ref{sec-oneavailable} and
\ref{sec-bothavailableFDR} will view the evidence from this SNP as
strong enough for it to be considered a replicated finding.

\section{Replicability analysis with a primary and a follow-up study}\label{sec-oneavailable}
For the family of $m$ no replicability null hypotheses
$H_{NR,1},\ldots,H_{NR,m}$, we consider two relevant error measures:
the probability that at least one no replicability null hypothesis
is falsely rejected, i.e. the FWER, and the expected fraction of
false rejections out of all rejections of no replicability null
hypotheses, that is the FDR.

\begin{procedure}\label{proc-FWER}
 The two stage FWER controlling procedure for testing the
family of no replicability null hypotheses with parameters
$(\alpha_1,\alpha)$, where $0<\alpha_1<\alpha<1$:
\begin{enumerate}
\item Let $\mathcal R_1$ be the set of indices of elementary hypotheses
that are selected for testing in a follow-up study based on the data
from the primary study.
\item Apply a FWER controlling
procedure at level $\alpha_1$, using the data from the primary study
only, on the family of null hypotheses $H_1,\ldots,H_m$, and let
$\mathcal R_{p}\subseteq \{ 1,\ldots,m \}$ be the set of indices of
rejected hypotheses. Apply a FWER controlling procedure at level
$\alpha-\alpha_1$, using the data from the follow-up study only, on
the family of selected null hypotheses $\{H_j: j\in \mathcal R_1
\}$, and let $\mathcal R_{f}\subseteq \mathcal R_1$ be the set of
indices of rejected hypotheses. Then the set of indices of rejected
no replicability null hypotheses is $\mathcal R_{f}\cap \mathcal
R_{p}$.
\end{enumerate}
\end{procedure}

\begin{theorem}\label{thm-fwer}
For two independent studies,  Procedure \ref{proc-FWER} controls the
FWER at level $\alpha$ for the family of no replicability null
hypotheses $H_{NR,1},\ldots,H_{NR,m}$.
\end{theorem}
\begin{proof}
Let $V_p = \sum_{j\in \mathcal R_p}(1-h_{1j})$ and $V_f = \sum_{j\in
\mathcal R_f}(1-h_{2j})$ be the number of true elementary null
hypotheses rejected, respectively, in the primary study and in the follow-up study. Then
\begin{eqnarray}
FWER  \leq E(\textbf{I}[V_p+V_f>0])\leq E(\textbf{I}[V_p>0])
+E(E(\textbf{I}[V_f>0]|p_1))\leq \alpha_1+\alpha-\alpha_1 = \alpha,
\nonumber
\end{eqnarray}
where the last inequality follows from the fact that $V_f$ is
independent of the data from the primary study, and that in both
studies a FWER controlling procedure is applied.
\end{proof}

 Using Bonferroni in Procedure \ref{proc-FWER} amounts to rejecting
$H_{NR,j}$ if $(p_{1j}, p_{2j})\leq (\alpha_1/m,
(\alpha-\alpha_1)/|\mathcal{R}_1|)$, for $j\in \mathcal{R}_1$.
Alternatively, the results can be reported in terms of
Bonferroni-replicability adjusted $p$-values $p^{Bonf-REPadj}_j =
\max \left(mp_{1j}/c, |\mathcal{R}_1|p_{2j}/(1-c)\right)$, where $c
= \alpha_1/\alpha$. Procedure \ref{proc-FWER} using Bonferroni is
equivalent to rejecting all hypotheses with Bonferroni-replicability
adjusted $p$-values at most $\alpha$.

The selection rule affects the power of Procedure \ref{proc-FWER}. A
natural choice for a selection rule is the set of rejected
hypotheses by the FWER controlling procedure at level $\alpha_1$ on
the primary study $p$-values, since the set of indices of rejected
no replicability null hypotheses is a subset of this set. A rule
that selects by the FWER controlling procedure at level $\alpha$ is
not as good, since any additional hypotheses selected will not be
rejected but will result in a more severe multiple testing problem
for the follow-up study. The choice of  $\alpha_1$ also affects the
power of Procedure \ref{proc-FWER}. We observed in simulations
(Supplementary Material) that although the optimal $\alpha_1$ varies
with effect size, the power function is quite flat as long as
$\alpha_1/\alpha$ is not too close to zero or one.

In many modern applications, controlling the FWER is unnecessary and
results in overly conservative inferences. In genomics research, it
is often enough to guarantee FDR control, see \cite{storey03} and
\cite{reiner2}, among others.

\begin{procedure}\label{procfdr}
 The two stage FDR controlling procedure for testing a family of no
replicability null hypotheses with parameters $(q_1,q)$, where
$0<q_1<q<1$:
\begin{enumerate}
\item Let $\mathcal R_1$ be the set of indices of elementary
hypotheses that are selected for testing in a follow-up study based
on the data from the primary study. Let $R_1 = |\mathcal R_1|$ be
the cardinality of this set.
\item Let $$R_2\triangleq\max\left\{r:
\sum_{j\in\mathcal{R}_1}\textbf{I}\left[(p_{1j},
p_{2j})\leq\left(\frac{rq_1}{m}, \frac{r(q-q_1)}{R_1}\right)\right]
= r\right\}.$$  Then the set of indices of rejected no replicability
null hypotheses is
$$\mathcal R_2= \left\{j: (p_{1j}, p_{2j})\leq\left(\frac{R_2q_1}{m},
\frac{R_2(q-q_1)}{R_1}\right), j \in \mathcal R_1\right\}.$$
\end{enumerate}
\end{procedure}

The results of Procedure \ref{procfdr} can be reported in terms of
FDR-replicability adjusted $p$-values. Let $c = q_1/q$,
\begin{align}\label{eq-FDRZvalue}
Z_j=\max\left(\frac{mp_{1j}}{c},\,\frac{R_1p_{2j}}{1-c}\right), j
\in \mathcal{R}_1,
\end{align}
and let $Z_{(1)}\leq \ldots \leq Z_{(R_1)}$ be the sorted
$Z$-values. Then the $i$th largest FDR-replicability adjusted
$p$-value is
\begin{align}\label{eq-FDRREPadjpvalue}
p_{(i)}^{REPadj} = \min_{j\geq i} \frac{Z_{(j)}}{j}.
\end{align}
Procedure \ref{procfdr} with parameters $(q_1, q) = (cq, q)$ is
equivalent to rejecting all no replicability null hypotheses with
FDR-replicability adjusted $p$-values at most $q$.

\begin{definition}\label{validsel}
 A \emph{valid selection rule} for step 1 of Procedure
\ref{procfdr} satisfies the following condition: for any $j\in
\mathcal{R}_1$, fixing all the $p$-values except for $p_{1j}$ and
changing $p_{1j}$ so that $H_{1j}$ is still selected, will not
change the set $\mathcal R_1$.
\end{definition}

 It is easy to see that this condition is satisfied if
$\mathcal R_1$ contains the smallest fixed number of $p$-values, all
hypotheses with $p$-value below a given threshold, or if $\mathcal
R_1$ contains the rejected indices from a BH procedure
on the $p$-values from the primary study.
Adaptive FDR procedures on the $p$-values from the primary study,
e.g. \cite{BHadaptive}, \cite{storey1}, \cite{yoav2}, and
\citet{Blanchard09}, are
non-valid selection rules.

\begin{theorem}\label{thmfdr}
If all the $p$-values are jointly independent and the selection rule
in step 1 of Procedure \ref{procfdr} is a valid selection rule, then
Procedure \ref{procfdr} controls the FDR at level $q$ for the family
of no replicability null hypotheses $H_{NR,1},\ldots,H_{NR,m}$.
\end{theorem}
See Appendix \ref{proof-thmfdr} for the proof.

The selection rule and the choice of $q_1$ affect the power of
Procedure \ref{procfdr}. A natural choice for a selection rule is
the set of rejected hypotheses by the BH procedure at level $q_1$ on
the primary study $p$-values, since the set of indices of rejected
no replicability null hypotheses  is a subset of this set. A rule
that selects by the BH procedure at level $q$ is not as good as the
rule at level $q_1$, since any additional hypotheses selected will
not be rejected, but will result in more severe thresholds on the
follow-up study $p$-values. In Figure \ref{figSelectionRule} we
showed in a simulated example that the BH procedure at level $q_1$
was very close to selecting the optimal number of hypotheses for
follow-up. We recommend using it when there are no additional
constraints that require choosing only a small number of
hypotheses for follow-up. The optimal choice of $q_1$ depends on
$|I_{00}|, |I_{01}|, |I_{10}|, |I_{11}|$, and the non-null
distribution of the $p$-values, and therefore guidelines for
choosing $q_1$ are application specific. In simulated GWAS in
Section \ref{sec-example} the choice of $q_1$ had little effect on
the average number of discoveries.

Theorem \ref{thmfdr} assumes independence of the $p$-values within
each study as well as  across the studies. However, the assumption
of independence among the $p$-values within each study may not be
realistic in many applications. Particularly, in GWAS there is
dependency across the SNPs, therefore the $p$-values within each
study may be dependent. \cite{yoav6} proved that the BH procedure
controls the FDR when the $p$-values have a special dependency
called PRDS.

\begin{definition}\label{defPRDS}\citep{yoav6} The set of $p$-values $P_1,\ldots,P_M$ has property PRDS if for any increasing set $D$, and for each true null hypothesis
$i$, $Pr((P_1,\ldots, P_M) \in D| P_i =p)$ is nondecreasing in $p$ .
\end{definition}

If the $p$-values are independent in the primary study, yet have
property PRDS in the follow-up study, Theorem S3.1 in the
Supplementary Material shows that the result in Theorem \ref{thmfdr}
holds. For arbitrary dependence among the $p$-values in the primary study,
  a modification of the cut-off level of Procedure \ref{procfdr}
will guarantee that the FDR is controlled at the nominal level. The most severe modification, that will guarantee FDR control  for any valid selection rule,
is to apply Procedure \ref{procfdr} with the modification in item 1 of Theorem \ref{genthm} below. However,  in  item 2 of Theorem \ref{genthm}  we show that the modification  factor may be  smaller than $\sum_{i=1}^m 1/i\approx \log m$  if the selected hypotheses for follow-up are a subset of the  hypotheses with primary study $p$-values below a fixed cut-off $t$.
For example, in GWAS it is common to select hypotheses with primary study $p$-values below $1/m$, where $m$ is the number of hypotheses in the primary study \citep{lander95}. If  $t\geq \frac{q_1}{1+\sum_{i=1}^{m-1}\frac 1i}$, then the modification in item 1 of Theorem \ref{genthm} cannot be improved. However, if  $t< \frac{q_1}{1+\sum_{i=1}^{m-1}\frac 1i}$, then the modification in item 2 of Theorem \ref{genthm} is less conservative than the modification in item 1.
For typical values of $q_1$ (e.g. $q_1 \in [0.005, 0.045]$) and large $m$, the threshold $t$ will often be below $ \frac{q_1}{1+\sum_{i=1}^{m-1}\frac 1i}$, and therefore item 2 may be useful in applications. Note, moreover, that if $t\leq q_1/m$,
then item 2 of Theorem \ref{genthm} states that no modification is required,  so for a valid selection rule which selects a subset of the set of hypotheses with primary study $p$-values below $t$, where $t\leq q_1/m$, Procedure \ref{procfdr} is valid  for any form of dependency among the $p$-values in the primary study.

\begin{theorem}\label{genthm}
 Assume that the follow-up study $p$-values have property PRDS, and are independent of the $p$-values in the primary study. Then Procedure \ref{procfdr}
 controls the FDR at level
$q$ for the family of no replicability null hypotheses
$H_{NR,1},\ldots,H_{NR,m}$ if the selection
rule used in step 1 of Procedure 3.2 is a valid selection rule, and the expressions in step 2 of Procedure \ref{procfdr} are modified as follows:
\begin{enumerate}
\item  In the terms $rq_1/m$ and $R_2q_1/m$ only, $q_1$ is replaced by $\widetilde{q_1} = q_1/(\sum_{i=1}^m1/i)$.
\item In the terms $rq_1/m$ and $R_2q_1/m$ only, $q_1$ is replaced by $\widetilde{q}_1$, where $$\widetilde{q}_1 = \max \{x:\,\, x(1+\sum_{i=1}^{\lceil
tm/x-1 \rceil}1/i)=q_1\},$$ if  only  hypotheses with primary study $p$-values at most a fixed threshold $t$ are considered for follow-up,  i.e.
$\mathcal{R}_1\subseteq\{j\in\{1,\ldots,m\}:\,P_{1j}\leq t\}$, where $t< \frac{q_1}{1+\sum_{i=1}^{m-1}\frac 1i}$.
\end{enumerate}
\end{theorem}

See Supplementary Material for the proof, as well as for additional results under dependency. Specifically, Theorem S3.2 in the Supplementary Material shows that in the more general setting of arbitrary dependence among the follow-up study $p$-values,  it is also necessary to replace
$(q-q_1)$ with $(q-q_1)/(\sum_{i=1}^{R_1} 1/i)$ in the terms $r(q-q_1)/R_1$ and $R_2(q-q_1)/R_1$
in expression 2 of Procedure \ref{procfdr}. These results  are similar to the
result in \cite{yoav6} for the BH procedure in their Theorem 1.3.

\begin{remark}\label{rem-Naive}
\cite{BY05} proved in their Proposition 3 that the procedure that
applies the BH procedure at level $q_1$ on the primary study
$p$-values, and the BH procedure at level $q-q_1$ on the follow-up
study $p$-values, controls the FDR  at level $q_1(q-q_1)<q$ on the
family of global null hypotheses, $H_{G1}, \ldots, H_{Gm}$, where
$H_{Gj}: (h_{1j}, h_{2j}) = (0,0)$. However, on the family of no
replicability null hypotheses, $H_{NR,1}, \ldots, H_{NR,m}$, the FDR
of this procedure may be higher than the nominal level $q$. The key
difference between Procedure \ref{procfdr} and such a two stage
procedure, is the requirement that the two $p$-values from a
selected hypothesis have to simultaneously be smaller than two
thresholds. In an extreme scenario where all hypotheses are from
$I_{10}$ or $I_{01}$, and the $p$-values from false null hypotheses
are zero, the two stage procedure may have an FDR of one, as
follows. The BH procedure on the primary study $p$-values will
reject all hypotheses from $I_{10}$ but also few from $I_{01}$ (when
$|I_{01}|$ and $|I_{10}|$ are large enough), and the hypotheses from
$I_{01}$ will be rejected by the BH procedure on the follow-up study
$p$-values, resulting in an FDR of one. However, Procedure
$\ref{procfdr}$ will have an FDR level below $q$. To see this, note
that in order to reject a no replicability null hypothesis by
Procedure \ref{procfdr}, the $p$-value of the Simes test
\citep{simes} for the intersection of the elementary hypotheses
indexed by $I_{01}$, using the data from the primary study, has to
be below $q_1$, or the Simes test $p$-value for the intersection of
elementary hypotheses indexed by $I_{10}\cap \mathcal{R}_1$, using
the data from the follow-up study, has to be below $q-q_1$.
Therefore, the probability of rejecting at least one no
replicability null hypothesis, which coincides with the FDR since
all no replicability null hypotheses are true, is at most $q$. See
Figure \ref{fig050508}, right panel, for a more realistic simulated
example.
\end{remark}

\section{Replicability analysis  with no division into primary and follow-up studies}\label{sec-bothavailableFDR}
Consider now a situation where both studies are available  before
the analysis. If some of the elementary hypotheses are examined in
only one of the studies, then these hypotheses are not considered
for replicability analysis. In this setting, there is no primary
study and follow-up study.  We propose the following generalization
of Procedure \ref{procfdr}, that can be tuned to treat the two
studies symmetrically. Without loss of generality, we label the
studies as study one and study two.
\begin{procedure}\label{procfdrsym}
The generalized two stage procedure for testing a family of no
replicability null hypotheses with parameters $(w_1,q_1,q)$, where
$0\leq w_1\leq 1$ and $0<q_1<q<1$:
\begin{enumerate}
\item Apply Procedure \ref{procfdr} with parameters $(w_1q_1,w_1q)$ with study one as the primary study and study two as the follow-up study. Denote the set of indices of
rejected no replicability null hypotheses by $\mathcal R_{12,
w_1q}$.
\item Reverse the roles of study one and study two. Apply Procedure \ref{procfdr}
with parameters $((1-w_1)q_1,(1-w_1)q)$. Denote the set of indices
of rejected no replicability null hypotheses by $\mathcal R_{21,
(1-w_1)q}$.
\item The set of indices of
rejected no replicability null hypotheses is $\mathcal R_{12,
w_1q}\cup \mathcal R_{21, (1-w_1)q}$.
\end{enumerate}
\end{procedure}

\begin{theorem}\label{thmprocfdrsym}
Procedure \ref{procfdrsym} controls the FDR at level $q$ for the
family of no replicability null hypotheses
$H_{NR,1},\ldots,H_{NR,m}$ if all $p$-values are jointly independent and  the selection rule
in step 1 of Procedure \ref{procfdr} is a valid selection rule.
\end{theorem}
See Appendix \ref{proof-thmprocfdrsym} for the proof.

Choosing $w_1=1$ results in Procedure \ref{procfdr}, where study one
has the role of the primary study and study two has the role of the
follow-up study. Similarly, choosing $w_1=0$ results in Procedure
\ref{procfdr} with the roles of study one and study two reversed.
The choice $0<w_1<1$ reflects the similarity of Procedure
\ref{procfdrsym} to Procedure \ref{procfdr} in the following way:
when $w_1$ is close to one (zero), Procedure \ref{procfdrsym} gives
similar results to Procedure \ref{procfdr} with study one (two) as
the primary study. The choice $w_1=0.5$ results in a variant of
Procedure \ref{procfdr} that is symmetric with respect to both
studies.

\section{GWAS examples} \label{sec-example}
In this section we demonstrate the suggested methods on two real data examples and on a GWAS simulation. A replicability analysis with FWER control is carried out for the first example, that has  only five hypotheses in the follow-up study. A replicability analysis with FDR control is carried out for the second example, that has 126 hypotheses in the follow-up study. Finally, in order to examine the robustness of procedure \ref{procfdr} for GWAS type dependency, examples were simulated that retained the dependencies in the data that occur in GWAS.

\paragraph{Example 1.} We reproduce in Table \ref{tab1}, columns 1-4, a subset of the
columns of Table 1 of results of \cite{Bis11}. We added in columns
5-7 the Bonferroni-replicability adjusted $p$-values for $c=q_1/q \in
\{0.2, 0.5, 0.8 \}$. Procedure \ref{proc-FWER} with parameters
$(q_1,q) = (0.025, 0.05)$  or $(q_1,q) = (0.04, 0.05)$ identified the SNP near MSRB3 as having
replicated association with the phenotype.
 The choice of $c$ should be made prior to analysis, and the
choice $c=0.8$ may be preferred over $c\leq 0.5$ when it is believed
that the power to detect an association in the primary study using a
threshold of order $1/(2.5\times10^6)$  is smaller than the power to
detect an association in the follow-up study using a threshold of
order $1/5$.

\begin{table}
\caption{The $p$-values of SNPs from the primary and follow-up
studies, from Table 1 of \cite{Bis11} (columns 3-4), and the
FDR-replicability adjusted $p$-values for various choices of
$c=q_1/q$ (columns 5-7).} \label{tab1}
\begin{tabular}{c|c|c|c|c|c|c|}
Locus &  Gene & Primary  & Follow-up   & \multicolumn{3}{|c|}{Bonferroni-replicability adjusted $p$-values}\\
&&study&study& $c=0.2$ & $c = 0.5$ & $c =0.8$\\
 \hline
 2q24 &
 DPP4 & $5.2\times 10^{-8}$ & 0.7 &   1.0000 & 1.0000 & 1.0000 \\
9q33  & ASTN2 & $1.0\times 10^{-7}$ & 0.2& 1.0000 & 0.5000 &0.3125 \\
12q14  & MSRB3 & $5.5\times 10^{-9}$ &0.002 & 0.06875 &  0.0275 & 0.0172\\
 & WIF1 & $2.2\times 10^{-8}$ & 0.0007& 0.2750 & 0.1100 & 0.0688 \\
12q24  & HRK & $4.8\times 10^{-8}$ & $5.8\times 10^{-5}$ & 0.6000 &
0.2400 & 0.1500
\end{tabular}
\end{table}

\paragraph{Example 2.} To discover
associations between SNPs and  Crohn's disease (CD), \cite{Barrett08} examined 635,547
SNPs on 3230 cases and 4829 controls of European descent,
  collected in three separate studies: NIDDK4, WTCCC5, and a Belgian-French study.
  The primary study $p$-values in this example are the meta-analysis $p$-values from the combined data from the three studies.
  Only hypotheses with primary study $p$-values below  $5\times 10^{-5}$ were considered for follow-up. Although 526 SNPs met the selection criterion, only a subset of 126 SNPs were followed up.
  These 126 $p$-values were the smallest two $p$-values in 63 distinct regions, so the selection rule is a valid selection rule.
Procedure \ref{procfdr} with $(q_1,q) = (0.04, 0.01)$
identified 36 SNPs. In Appendix D, Table \ref{tab-crohn}  shows the $p$-values from the primary and follow-up studies, as well as the FDR-replicability adjusted $p$-values for the choice $c=0.8$, for these 36 replicability discoveries.
Since the $p$-values are not independent within each study, a more conservative analysis approach is to modify the cut-offs as suggested by Theorem \ref{genthm}.  Assuming  PRDS type dependency in the follow-up study, item 1 of Theorem \ref{genthm} suggests using $\tilde{q}_1 = 0.04/(\sum_{i=1}^{635,547}1/i) = 0.0029$ for the primary study cut-offs, while the follow-up study cut-offs remain unchanged. The modified procedure identified 21 SNPs. Column 7 of Table \ref{tab-crohn} shows the FDR-replicability adjusted $p$-values for the choice $c = 0.8$, where the adjustment is made as described in expressions (\ref{eq-FDRZvalue}) and (\ref{eq-FDRREPadjpvalue}), with $p_{1j}$ replaced by $\tilde{p}_{1j} = (\sum_{i=1}^{635,547}1/i)p_{1j} = 13.94\times p_{1j}
, j=1,\ldots,635,547$.
Since the SNPs considered for follow-up were only SNPs with primary study $p$-values below $5\times 10^{-5}$, one could use a less conservative procedure suggested in item 2 of Theorem \ref{genthm}, with $\tilde{q}_1 = 0.0038$, where $0.0038$ is the solution to $0.04 = x(\sum_{i=1}^{\lceil 635,547\times 5\times 10^{-5}/x-1 \rceil}1/i+1)$.
This procedure  resulted in 23 replicability discoveries.  The latter procedure  is the recommended procedure, if the investigator is not willing to assume that Procedure \ref{procfdr} is robust to deviations from independence within the primary study.
However, simulations in the next example suggest that for the type of dependencies that occur in GWAS, Procedure \ref{procfdr} may actually be conservative.  We come back to  the issue of robustness of Procedure \ref{procfdr} in the Discussion Section \ref{sec-discussion}.

\paragraph{GWAS simulation example.}
We simulated two GWAS from the simulator HAPGEN2
\citep{Su11}. The two studies were generated from two samples of the
HapMap project \citep{HapMap03}, a sample of 165 Utah residents with
Northern and Western European ancestry (CEU), and a sample of 109
Chinese in Metropolitan Denver, Colorado (CHD). In the
 CEU and CHD populations, respectively, 34 and 38 SNPs were set as disease SNPs with an increased multiplicative
relative risk of 1.2, and 18 of the disease SNPs were common to both
populations. Each study contained 4500 cases and 4500 referents. The
linkage disequilibrium (LD) across SNPs, as measured for the samples
in the HapMap project, was retained. Due to LD, the number of SNPs
associated with the phenotype in each study was larger than the
number of disease SNPs. In order to identify the SNPs in each study
that are truly associated with the phenotype, the simulation of 4500
cases and 4500 controls from the population was repeated 11 times,
and 11 $p$-values were produced per SNP. SNPs with Fisher's combined
$p$-value \citep{loughin} below the Bonferroni threshold were
considered to be truly associated with the disease. Our ground truth
included 1355 and 1010 SNPs associated with the disease in the CEU
and in the CHD population, respectively, out of which 274 SNPs were
associated with the disease in both populations.

As a standard preprocessing step, we removed SNPs with minor allele
frequency below 0.05, and thus the number of SNPs in the analysis
was reduced from 1,387,466 to 887,362, on average, for the 11 pairs
of studies. Our selection rule for Procedure \ref{procfdrsym} with
parameters $(w_1, q_1, q)$ was the BH procedure at level $w_1q_1$
when the primary study was the CEU study, and at level $(1-w_1)q_1$
when the primary study was the CHD study, since the potential set of
SNPs to be discovered as having replicated associations is at most
the set of SNPs that are discovered by the BH procedure (as
discussed in Section 3).
  Table \ref{tab3} presents the average
number of replicated findings, as well as the average false
discovery proportion (FDP) for the methods compared. The standard
error (SE) is presented in parentheses.  From rows 1 and 2 we see
that if there is no division into primary and follow-up studies,
then the symmetric Procedure \ref{procfdrsym} discovers more SNPs
with replicated associations than the BH procedure on maximum
$p$-values, while maintaining a low FDP. From rows 3-5, and 6-8, we
see that the choice of which study was the primary study had a large
effect on the average number of
discoveries, and  the choice of $q_1$ mattered little.

\begin{table}[ht]
\caption{For 4500 cases and 4500 referents in both studies, the
average number of associated and disease SNPs discovered (SE), and
the average FDP (SE), for different procedures. The selection rule
for Procedure \ref{procfdrsym} was the BH procedure at level
$w_1q_1$ when the CEU study was the primary study, and at level
$(1-w_1)q_1$ when the CHD study was the primary study. }\label{tab3}
\begin{center}
\begin{tabular}{|l|c|c|c|}
  \hline
   Procedure & \multicolumn{2}{c}{\# Replicated findings} &
FDP\\
&associated SNPs (SE) & disease SNPs (SE) & (SE)
\\ \hline
BH on maximum $p$-values & 29.182 (3.205) & 7.364 (0.432) & 0.000 (0.000) \\
 \ref{procfdrsym} with $w_1 = 0.5, q_1 = 0.025, q=0.05$ & 77.727 (6.378) & 11.455 (0.366) & 0.011 (0.005) \\
  \ref{procfdrsym} with $w_1 = 1, q_1 = 0.01, q=0.05$  & 74.091 (6.748) & 10.364 (0.310) & 0.012 (0.006) \\
  \ref{procfdrsym} with $w_1 = 1, q_1 = 0.025, q=0.05$ & 76.091 (6.221) & 10.727 (0.359) & 0.012 (0.005) \\
  \ref{procfdrsym} with $w_1 = 1, q_1 = 0.04, q=0.05$ & 69.545 (5.745) & 10.818 (0.352) & 0.009 (0.005) \\
  \ref{procfdrsym} with $w_1 = 0, q_1 = 0.01, q=0.05$ & 35.545 (4.575) & 7.364 (0.607) & 0.008 (0.008) \\
  \ref{procfdrsym} with $w_1 = 0, q_1 = 0.025, q=0.05$ & 41.455 (5.294) & 8.273 (0.469) & 0.007 (0.007) \\
  \ref{procfdrsym} with $w_1 = 0, q_1 = 0.04, q=0.05$ & 42.273 (4.158) & 8.545 (0.312) & 0.000 (0.000) \\
   \hline
\end{tabular}
\end{center}
\end{table}

From the last column in Table \ref{tab3} we see that the average FDP
was far below 0.05, suggesting that the procedures are conservative.
This conservatism can be alleviated if the following oracle
information were known: the fraction of SNPs with no association
with the phenotype in both studies, $f_{00}$, and  with association
with the phenotype only in the follow-up study, $f_{01}$. Then it
was possible to perform Procedure 4.1 at level $(w_1, q', 2q')$,
where $q'$ is the solution to $f_{00}(q')^2+(f_{01}+1)q' = q$ for
$w_1\in \{0,1\}$, and the solution to
$f_{00}(0.5q')^2+(f_{01}+1)0.5q' = 0.5q$ for $w_1=0.5$, with the
same guarantee of FDR control at level $q$, as follows from Appendix
\ref{proof-thmfdr-oracle}. Specifically, in our simulation $f_{00} =
0.9990$, $f_{01} =0.00036$ on average, after preprocessing. For FDR
control at level $q=0.05$, on average $q'=0.048$ for $w_1=0,1$ and
$q'=0.049$ for $w_1=0.5$. Table \ref{tab3b} shows the average FDP
and average number of rejections for Procedure \ref{procfdrsym} with
and without the oracle. Although the average FDP is higher with the
oracle, it is still below the nominal 0.05 level for two main
reasons. First, our simulation preserves the LD pattern of the SNPs,
and thus the $p$-values within each study are not independent.
Second, the upper bound of $f_{00}(q')^2+(f_{01}+1)q'$ is not a
tight upper bound for the actual FDR level.  A tighter oracle upper
bound requires knowing the expectation of $|{\mathcal R}_1 \cap
I_{10}|/|{\mathcal R}_1|$, and this bound is tight if the non-null
effect sizes in $I_{10}\cup I_{01}$ are extremely large.

\begin{table}[ht]
\caption{The average FDP and average number of rejections  for
Procedure \ref{procfdrsym} with and without the oracle, for FDR
control at level 0.05. }\label{tab3b}
\begin{center}
\begin{tabular}{|l|c|c|c|c|}
  \hline
    &  \multicolumn{2}{|c|}{FDP} &
 \multicolumn{2}{|c|}{\# Replicated findings}\\
 & Oracle  & $(q_1,q)=(0.025, 0.05)$ &  Oracle & $(q_1,q) = (0.025,
0.05)$
\\ \hline
 $w_1 = 0.5$ & 0.023 & 0.011 & 90 & 78  \\
$w_1 = 1$ & 0.023 & 0.012 &  85 & 76  \\
$w_1 = 0$ & 0.029 & 0.007 & 50 & 41  \\
   \hline
\end{tabular}
\end{center}
\end{table}

For the two studies from the CEU and CHD populations, a
meta-analysis was performed by first combining the SNP $p$-values
using Fisher's combining method, and then applying the BH procedure
at level 0.05 on the combined $p$-values. The average number of SNPs
associated with the disease in at least one study was 393, while
less than 80 SNPs were discovered to have replicated associations
(Table \ref{tab3}). The two main reasons for discovering more SNPs
in a typical meta analysis are as follows. First, the simulation
setting contained five times more associated SNPs than SNPs with
replicated associations. Second, for a SNP with a replicated
association, the power to detect that the association is replicated
is lower than the power to detect that there is an association in at
least one study.
The discovered SNPs with replicated associations were a subset of
the discovered associated SNPs, but their meta-analysis $p$-values
were not ranked smallest among all meta-analysis $p$-values (not
shown). Importantly, the discoveries from the meta-analysis could
not serve as evidence towards replicability, since while the average
fraction of SNPs with no association in both studies among the
meta-analysis discoveries was 0.06, the average fraction of SNPs
with no replicated association among the meta-analysis discoveries
was 0.78.

\section{A simulation study} \label{sec-sim}
The goal of the simulations was threefold. First, to investigate the
effect of the choice of $q_1$ and $w_1$ on the power of Procedures
\ref{procfdr} and \ref{procfdrsym}. Second, to compare these
procedures to the alternative of applying BH on the maximum
$p$-values, i.e. the partial conjunction approach when exactly two
studies are analyzed. Third, to investigate the effect of the
selection rule on the power of the procedures.

The procedures compared were (1) the BH procedure at level 0.05 on
maximum $p$-values; (2) Procedure \ref{procfdrsym} with $w_1 \in
\{0, 0.5, 1\}$, $c = q_1/q \in \{0.1, 0.2, \ldots,0.9\}$, and
$q=0.05$; and (3) the naive (BH-$i$, BH-$j$) procedure, $i,j\in
\{1,2\}$, $i\neq j$, which applies the BH procedure at level 0.05 on
the $p$-values of study $i$, and separately on the $p$-values of study
$j$ for the hypotheses that were rejected in study $i$, and
declares hypotheses rejected in both studies as false no replicability null hypotheses;
(4) the oracle
 Procedure \ref{procfdr} with parameters $(q_1,q) = (q', 2q')$, where $q'$ was the solution to
$\frac{|I_{00}|}m (q')^2+\left(\frac{|I_{01}|}m+1\right)q' = 0.05$.
This oracle procedure controls the FDR at level $0.05$, see Appendix
\ref{proof-thmfdr-oracle} for a proof.

The $p$-values were generated independently as follows. For $H_j$,
$j=1,\ldots, m$, $P_{1j} =
1-\Phi\left(\frac{X_{1j}}{\sigma_1}\right)$ and $P_{2j} =
1-\Phi\left(\frac{X_{2j}}{\sigma_2}\right)$, where $X_{1j}\sim
N(\mu_{1j}, \sigma_1^2)$ and $X_{2j}\sim N(\mu_{2j}, \sigma_2^2)$.
We let $\mu_{ij} = 0\cdot (1-h_{ij})+ \mu_i \cdot h_{ij}$, where
$i\in \{1,2\}$, and $ \mu_i \in \{ 0.5, 1, \ldots, 5\}$. We set
$m=1000$, and $f_{ij}  = |I_{ij}| /m$ for $i,j \in \{0,1\}$ as
follows: $f_{00}=0.9$, $f_{11}=0.1$; $f_{00}=0.9$,
$f_{01}=f_{10}=0.025$, $f_{11}=0.05$; $f_{01}=f_{10}=0.5$;
$f_{00}=0.8$, $f_{01}=f_{10}=0.1$. The standard deviations
$\sigma_1$ and $\sigma_2$ were either fixed values $\sigma_i\in
\{0.3, 1 \}$, $i\in \{ 1,2\}$, or reflected the fraction of sample
size allocated to the first study: $\sigma_1 = \sigma/\sqrt{\zeta
N}$, $\sigma_2 =\sigma/\sqrt{(1-\zeta)N}$, $\sigma=10$, $\zeta\in
\{0.1,0.2,\ldots,0.9 \}$, $N=1000$.

The simulation results were based on 1000 repetitions. The FDR was
estimated by averaging the FDP. The average power was estimated by
 the average number of  rejected false no replicability null
hypotheses, divided by $mf_{11}$.

\subsection{Simulation results}\label{subsec-simresults}
As expected from our theoretical results, in all the settings
considered the estimated FDR was below 0.05 for all procedures but
the naive (BH-$i$, BH-$j$) procedure.  The SE of the estimated FDR
and power were of the order of $10^{-3}$ for all  procedures under
all configurations considered.

Figure \ref{fig20025} compares the  power of the BH procedure on
maximum $p$-values, (1) above, and Procedure \ref{procfdrsym} with
$w_1 \in \{0, 0.5, 1\}$, $q_1 \in \{0.01, 0.025,0.04\}$, (2) above,
in a configuration with parameters $\sigma_1 = 0.3, \sigma_2 = 1,
f_{00} = 0.9, f_{01} =f_{10} = 0.025, f_{11} = 0.05$. The oracle
Procedure \ref{procfdr}, where the primary study is study one with
$\sigma_1 = 0.3$, is also examined. For each procedure the estimated
power and FDR  is shown as a function of the common expectation
under the alternative, $\mu = \mu_1=\mu_2$.
 Procedure \ref{procfdrsym} with $w_1=1$ is more
powerful than with $w_1=0.5$ or $w_1=0$, while the choice $w_1=0$ is
the worst in terms of power of Procedure \ref{procfdrsym}. Moreover,
Procedure \ref{procfdrsym} with $w_1\in\{0.5, 1\}$ is more powerful
than the BH procedure on maximum $p$-values. These findings were
consistent across all configurations of $f_{00}, f_{10}, f_{01},
f_{11}$ examined, when $\sigma_1 = 0.3$ and $\sigma_2=1$. Since the
oracle Procedure \ref{procfdr} and the BH procedure on maximum
$p$-values do not depend on $q_1$, their power curves are
 the same in figures (a), (b), and (c). We see that Procedure
\ref{procfdrsym} with $w_1=1$ is a close second to the oracle when
$q_1$ is 0.01 but is farther from the oracle as $q_1$ increases.
Similarly, the power of Procedure \ref{procfdrsym} with $w_1=0.5$
decreases as $q_1$ increases. However, Procedure \ref{procfdrsym}
with $w_1=0$ has largest power for $q_1=0.04$, and the least power
for $q_1=0.01$. These results are reasonable since the $p$-values of
study one tend to be much smaller than the $p$-values of study two
when the no replicability null hypotheses are false.
In Table \ref{tab-sim} we see that if the $p$-value distribution of
false no replicability null hypotheses is the same across studies,
then the optimal choice of $q_1$ is $q_1>q/2$. For example, when
$\mu = \mu_1=\mu_2 =2$ (row 2), the power is 0.65 with $q_1 =
0.005$, 0.77 with $q_1 = 0.045$, and the maximum power is 0.81 with
$q_1 = 0.035$.
\begin{table}[ht]
\begin{center}
\caption{The power of Procedure \ref{procfdr} with parameters
$(0.05c, 0.05)$ and the BH selection rule at level $0.05c$, for
different values of $\mu=\mu_1=\mu_2$,
 with $\sigma_1=\sigma_2 = 0.5$, $f_{00} =
0.9,  f_{01}=f_{10} = 0.025, f_{11}=0.05$. The optimal value of $c$
is in bold.}\label{tab-sim}
\begin{tabular}{|c|ccccccccc|}
  \hline
 & \multicolumn{9}{|c|}{$c = q_1/0.05$}\\ \hline
    $\mu$ & 0.1 & 0.2 & 0.3 & 0.4 & 0.5 & 0.6 & 0.7 & 0.8 & 0.9 \\
  \hline
1.5 & 0.143 & 0.195 & 0.224 & 0.245 & 0.257 & {\bf 0.258} & 0.248 & 0.226 & 0.181 \\
2.0 & 0.646 & 0.718 & 0.755 & 0.778 & 0.794 & 0.803 & {\bf 0.805} & 0.800 & 0.769 \\
2.5 & 0.934 & 0.955 & 0.965 & 0.971 & 0.975 & 0.977 & {\bf 0.978} & {\bf 0.978} & 0.974 \\
   \hline
\end{tabular}
\end{center}
\end{table}

Figure \ref{fig-fixedmu} compares the procedures (1) and (2) above
for the same configuration of $f_{ij}$, but for fixed $\mu = \mu_1 =
\mu_2$ and varying sample size of the two studies. The varying power
is described by the fraction $\zeta$ of sample allocated to the
first study. For the symmetric procedures, we see that for
$\zeta=0.1$ the power is the lowest, and it increases to reach its
maximum for equal allocation $\zeta=0.5$. Procedure \ref{procfdrsym}
with $w_1=0.5$  dominates the BH procedure on the maximum two study
$p$-values. For Procedure \ref{procfdrsym} with $w_1=1$, the maximum
is reached for $\zeta>0.5$.
It is the  most powerful of the three procedures examined for
$\zeta>0.6$.

In Figure \ref{fig050508} we consider the FDR level of Procedure
\ref{procfdrsym} with $w_1 \in \{0, 0.5,1 \}$, as well as of the
naive procedure in the null setting, where all no replicability null
hypotheses are true (i.e. $f_{11}=0$). The estimated FDR of (BH-$i$,
BH-$j$) procedure exceeds 0.05 in the settings where
$f_{10}=f_{01}=0.5$ and $f_{00}=0.8, f_{10}=f_{01}=0.1$. In these
settings the estimated FDR of both (BH-$1$, BH-$2$) and (BH-$2$,
BH-$1$) procedures are increasing functions of $\mu= \mu_1 = \mu_2$,
reaching one in the setting where $f_{10}=f_{01}=0.5$ (left), and
0.4 in the setting where $f_{00}=0.8, f_{10}=f_{01}=0.1$ (right).
Clearly, procedure (BH-$i$, BH-$j$) is not valid since it may be far
too liberal in terms of FDR level.

Finally, we examined how the selection rule affects the power.  In
Figure \ref{figSelectionRule} we show the  power as a function of
$\mu_1$ for Procedure \ref{procfdrsym} with parameters $w_1=0.5, q_1
= 0.025, q=0.05$,  for the following selection rules: BH at level
$0.0125$; the rule that selects the hypotheses with $k$ smallest
primary study $p$-values, where $k\in \{25,30,\ldots,100 \}$. The
remaining parameters were: $f_{00} = 0.9, f_{01} =f_{10} = 0.025,
f_{11} = 0.05, \sigma_1 =0.5, \sigma_2 = 1, \mu_2 = 3$.  For
different values of $\mu_1$ the optimal $k$ is different, and using
the BH procedure for selection is optimal for the entire range of
$\mu_1$.

\begin{figure}[!tpb]\
\centering
\subfloat[$q_1$=0.01]{\label{fig20025:3}\includegraphics[width=0.33\textwidth,
height=5.5cm]{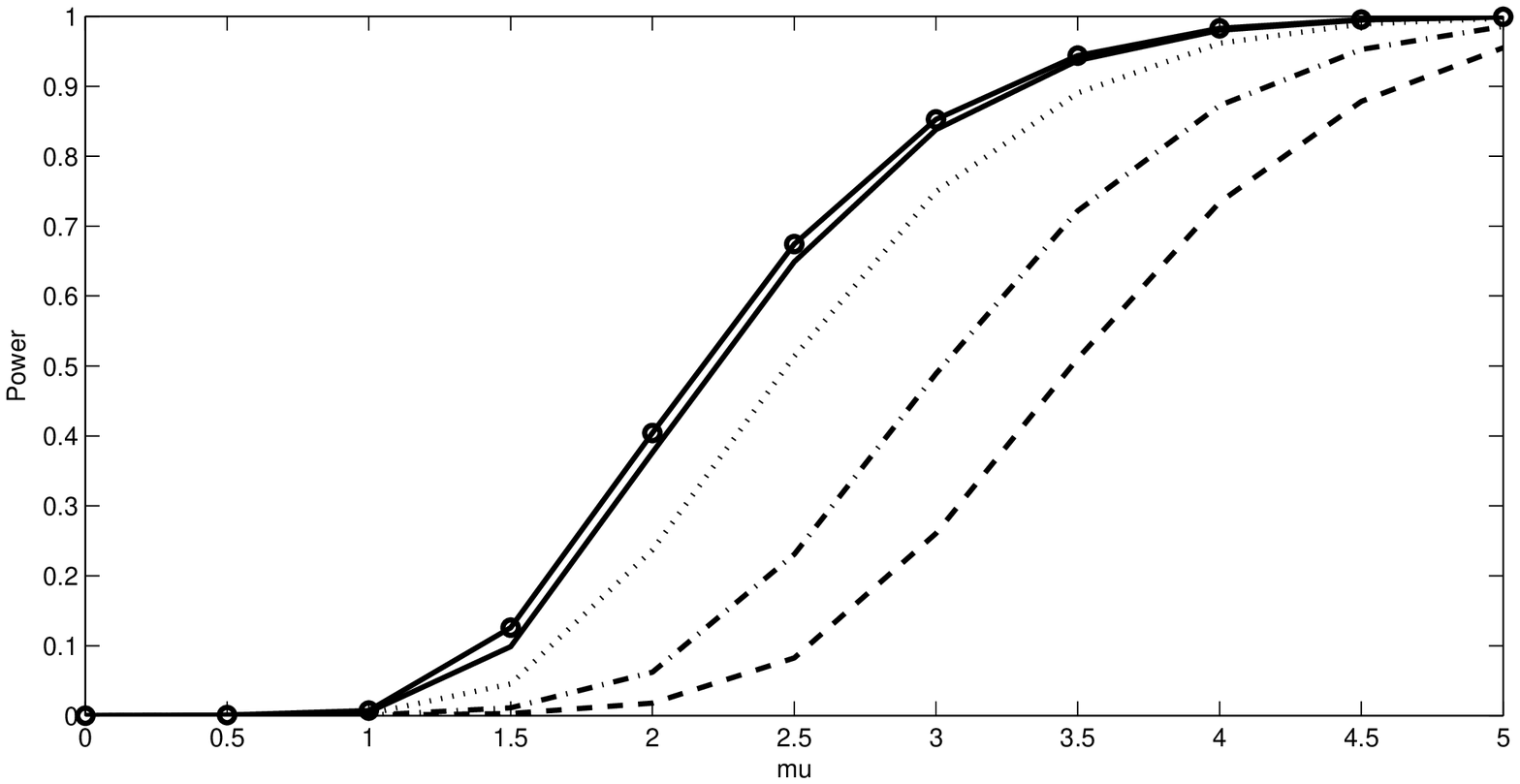}}
\subfloat[$q_1$=0.025]{\label{fig20025:1}\includegraphics[width=0.33\textwidth,
height=5.5cm]{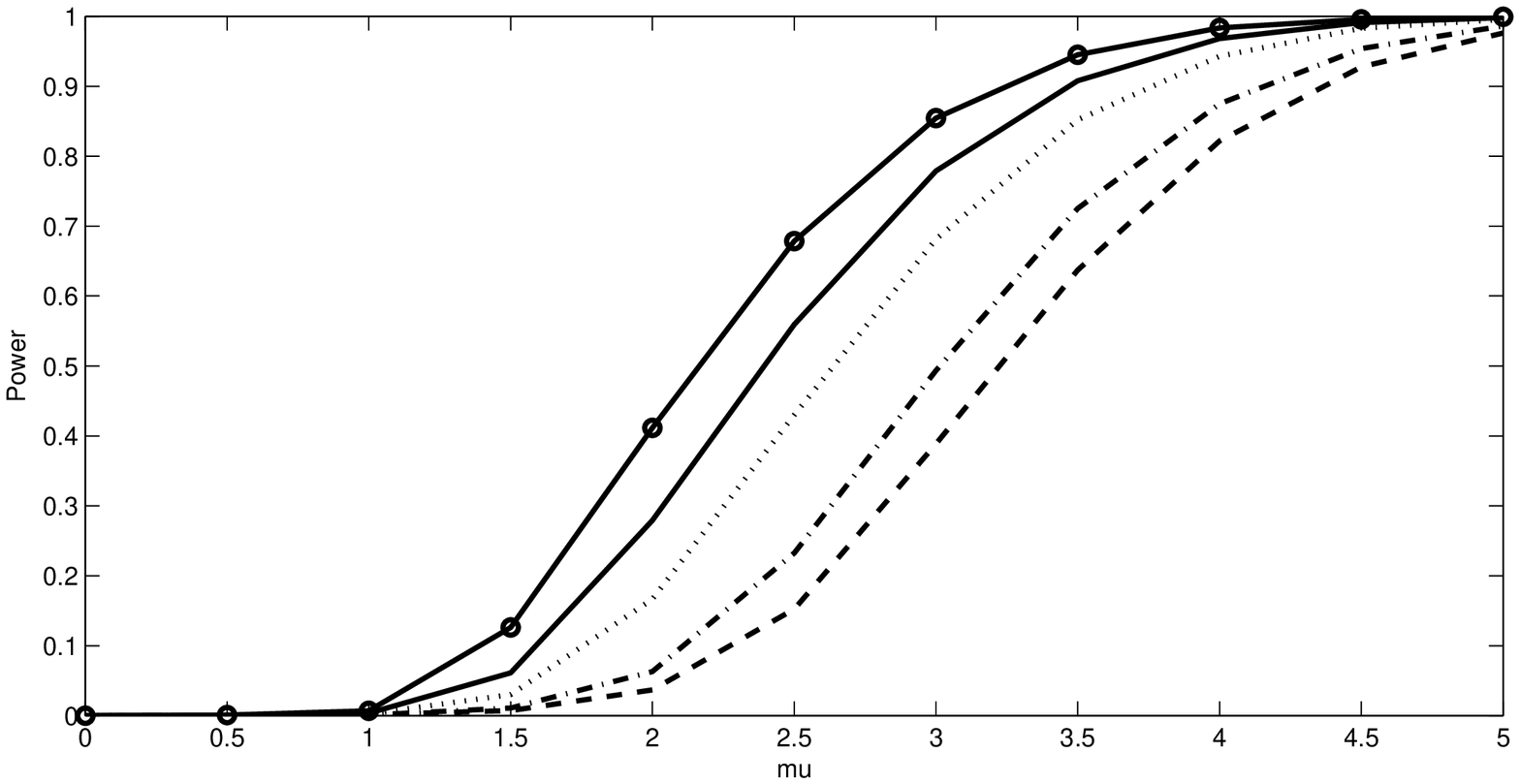}}
\subfloat[$q_1$=0.04]{\label{fig20025:2}\includegraphics[width=0.33\textwidth,
height=5.5cm]{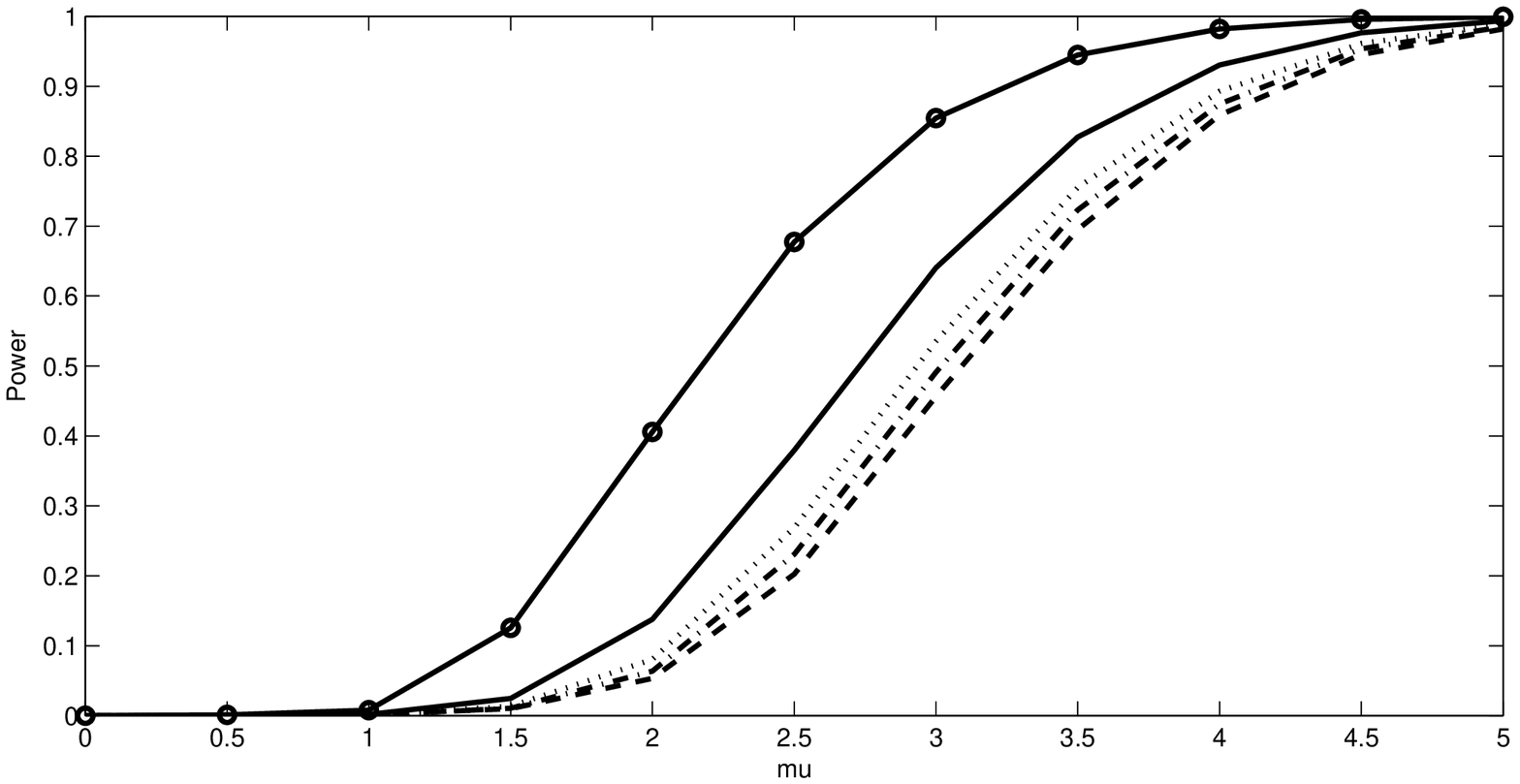}}
 \caption{Power as a function of  $\mu=\mu_1=\mu_2$, for $q_1$ of (a) 0.01, (b) 0.025, and (c) 0.04, using the following procedures: the oracle Procedure \ref{procfdr} (solid with
circles); the BH procedure at level 0.05 applied on maximum $p$-values
(dash-dotted);
 Procedure
\ref{procfdrsym} at level 0.05 with $w_1=0$ (dashed),  $w_1=0.5$
(dotted), and $w_1=1$ (solid), where the selection rule in steps 1
and 2 is the BH procedure at levels $w_1q_1$ and $(1-w_1)q_1$,
respectively.  The remaining parameters were $f_{00} = 0.9, f_{01} =
0.025, f_{10}=0.025, f_{11} = 0.05$, $\mu_1 = \mu_2 = \mu$,
$\sigma_1= 0.3$ and $\sigma_2=1$. }\label{fig20025}
\end{figure}
\begin{figure}[!tpb]\
\centering
\subfloat[$\mu=2$]{\label{fig20025:4}\includegraphics[width=0.48\textwidth,
height=6cm]{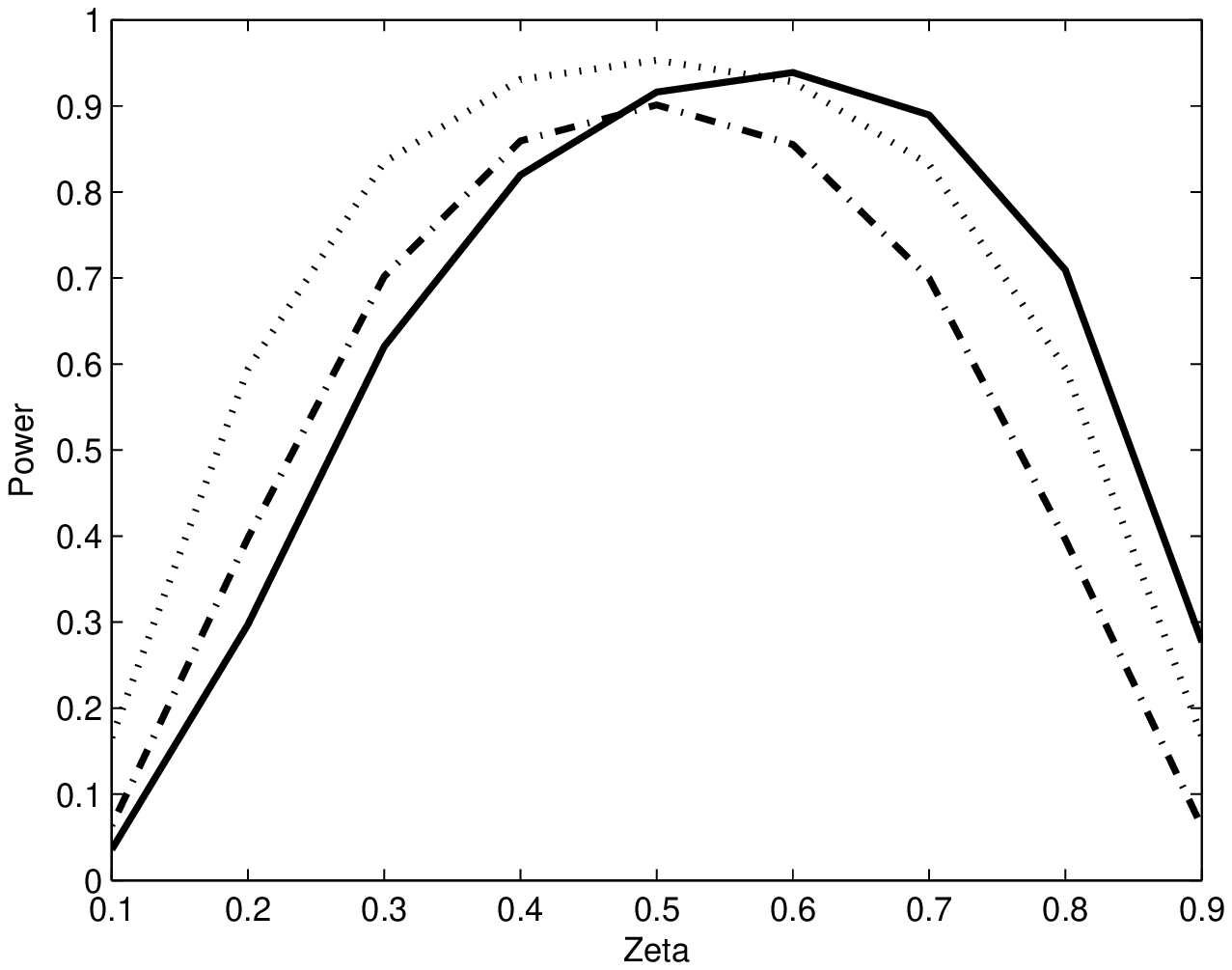}}\hspace{2pt}
\subfloat[$\mu=3$]{\label{fig20025:5}\includegraphics[width=0.48\textwidth,
height=6cm]{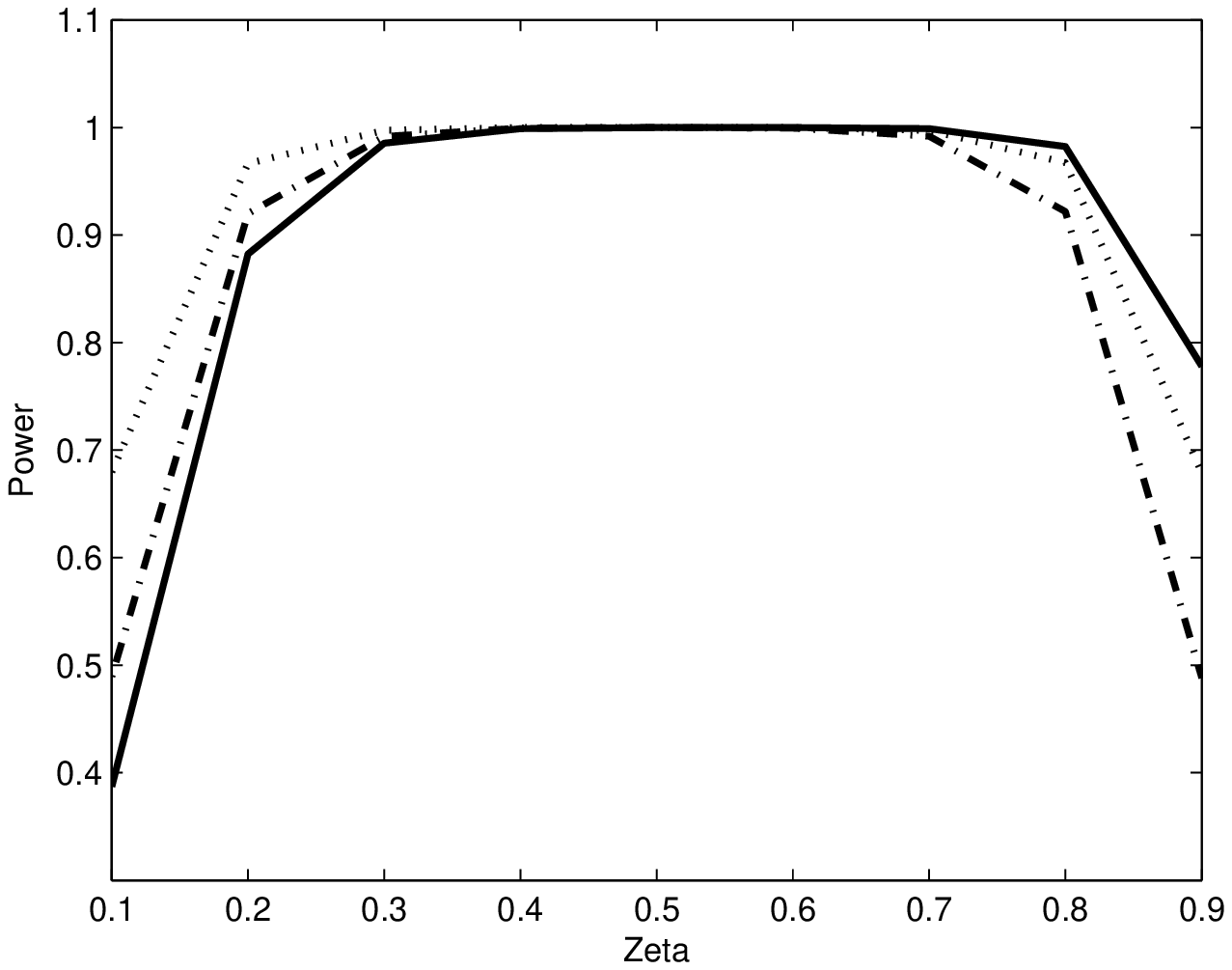}} \caption{Power as a function of fraction
$\zeta$ of sample size allocated to the primary study, for (a)
$\mu_1 = \mu_2 = 2$, and (b) $\mu_1=\mu_2=3$, for
 Procedure \ref{procfdrsym} with $w_1 = 1$ (solid), with $w_1=0.5$
(dotted), and for the BH procedure on the maximum of two studies
$p$-values (dash-dotted) at level $q=0.05$. The remaining parameters
were $f_{00} = 0.9, f_{01} = 0.025, f_{10}=0.025, f_{11} = 0.05$,
sample size $N=1000$, standard deviation $\sigma=10$.
}\label{fig-fixedmu}
\end{figure}

\begin{figure}[!tpb]
\centering
\includegraphics[width=0.48\textwidth, height=6cm]{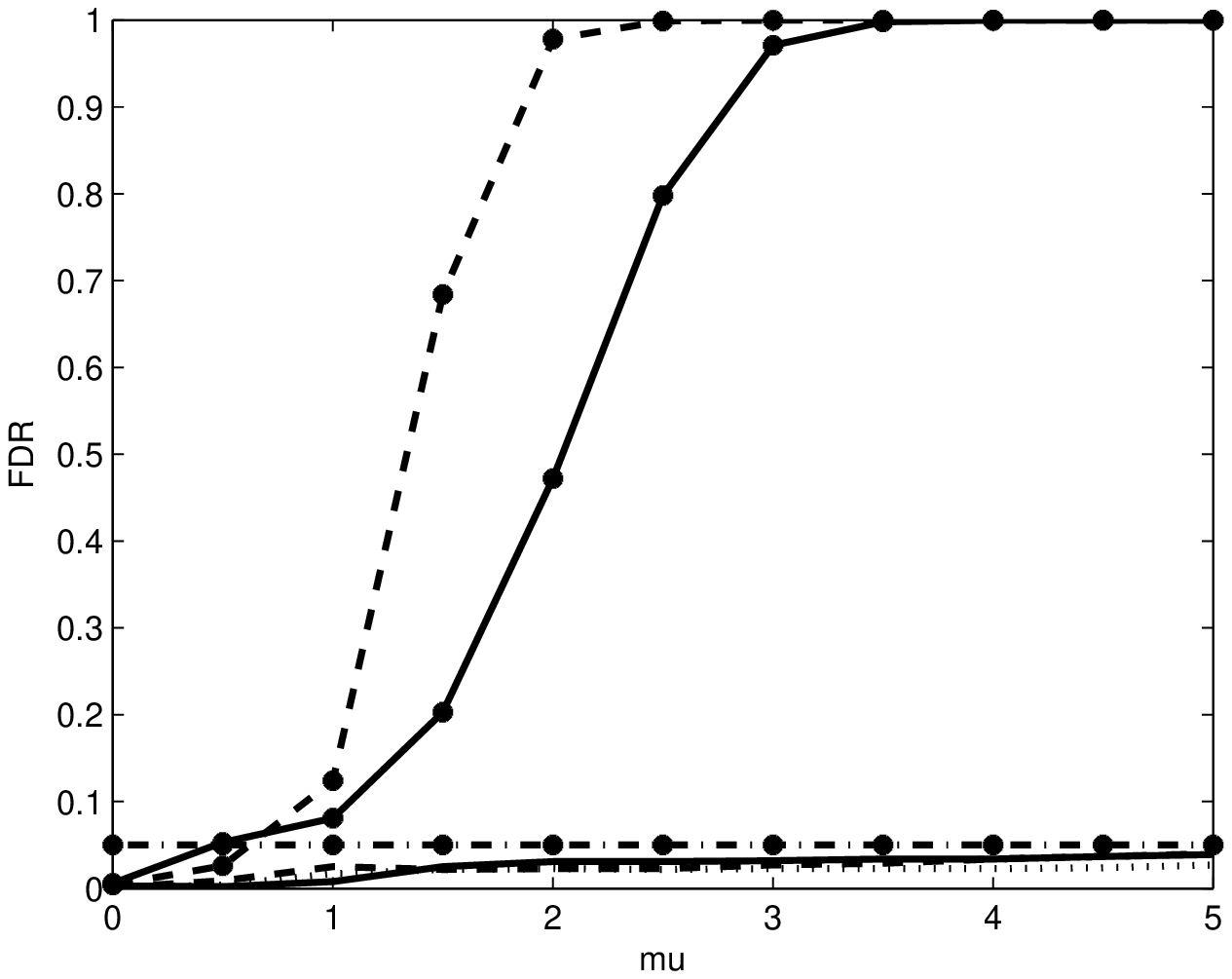}
\includegraphics[width=0.48\textwidth, height=6cm]{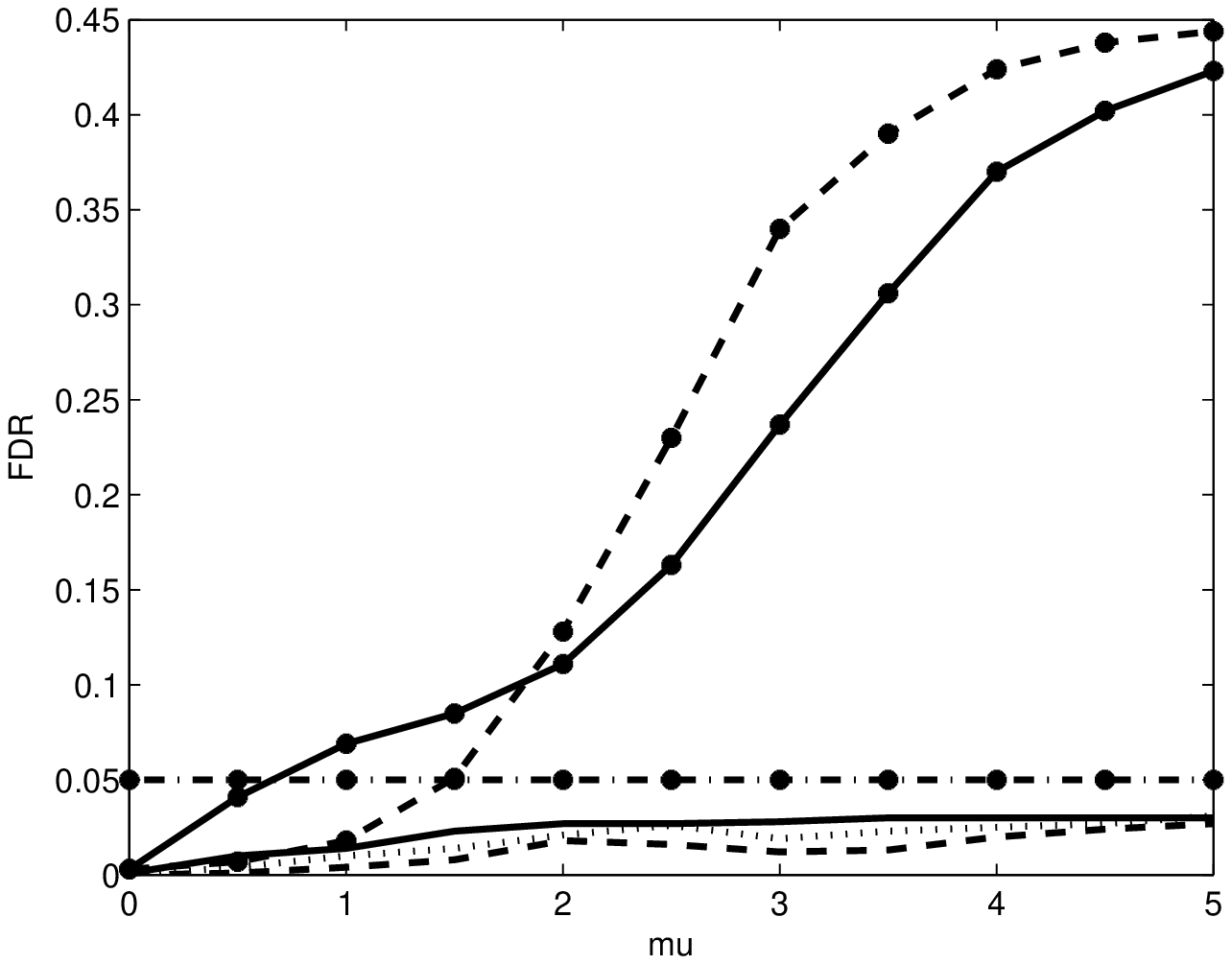}
\caption{FDR   versus $\mu=\mu_1=\mu_2$ for $f_{01}=f_{10}=0.5$
(left) and $f_{00}=0.8, f_{01}=f_{10}=0.1$ (right), for the
following procedures at level $q=0.05$: BH-1, BH-2 (solid with
circles); BH-2, BH-1 (dashed with circles);
 Procedure \ref{procfdrsym} with $q_1=0.025$ and $w_{1}$
of 1 (solid), 0.5 (dotted), or 0 (dashed). The standard deviations
were $\sigma_1 = 0.3$ and $\sigma_2 = 1$.  }\label{fig050508}
\end{figure}

\begin{figure}[!tpb]
\centering
\includegraphics[width=1\textwidth,
height=6cm]{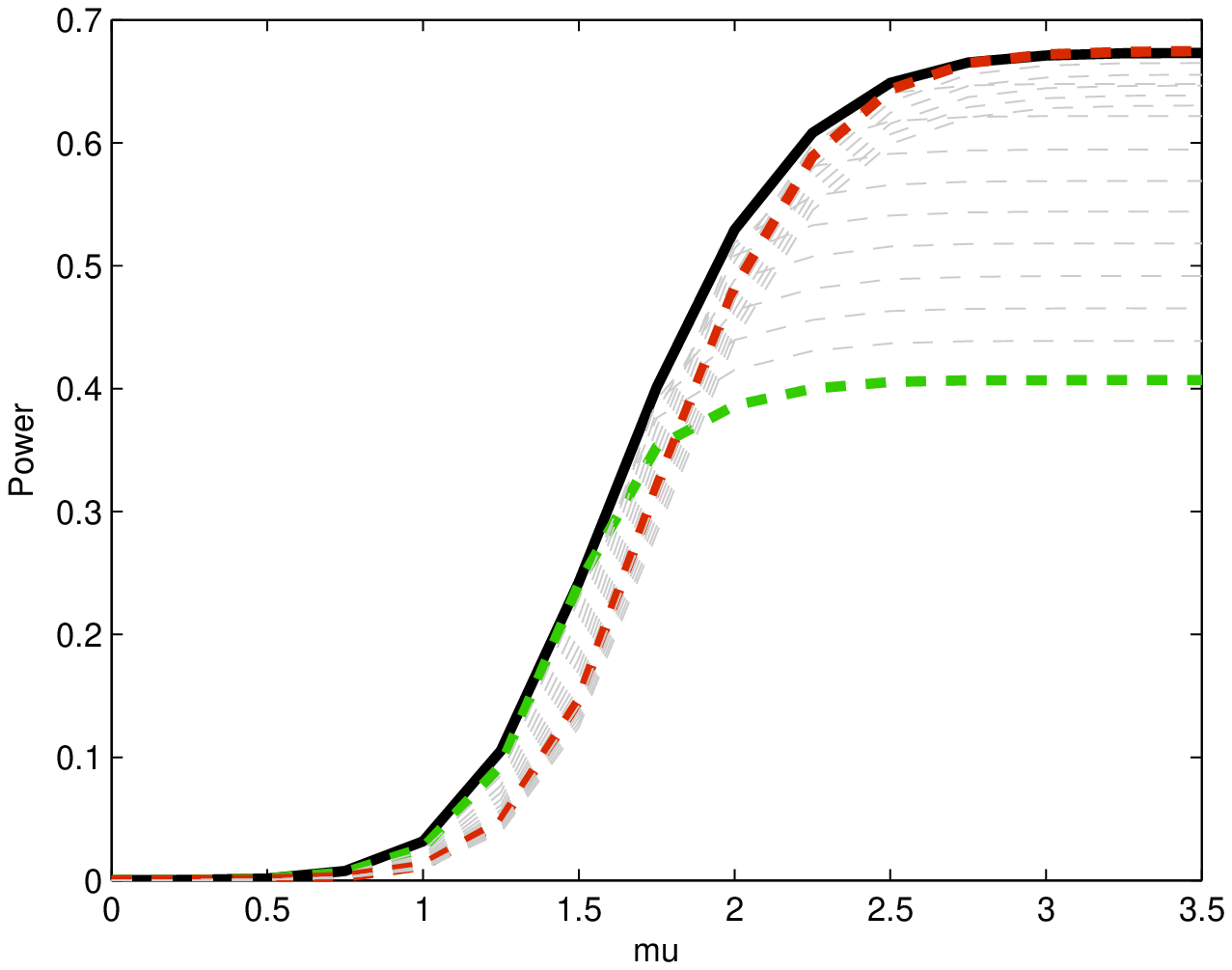} \caption{Power as a function of
$\mu_1$ for Procedure \ref{procfdrsym} with parameters $w_1=0.5, q_1
= 0.025, q=0.05$ for the following selection rules: BH at level
$0.0125$ (solid black curve); selection of the hypotheses with $k$
smallest primary study $p$-values, where $k=25$  (dashed green
curve), $k=75$ (dashed red curve), $k\in \{30,35,\ldots,100 \}$
(dashed grey curves). The remaining parameters were: $f_{00} = 0.9,
f_{01} =f_{10} = 0.025, f_{11} = 0.05, \sigma_1 =0.5, \sigma_2 = 1,
\mu_2 = 3$. }\label{figSelectionRule}
\end{figure}

\section{Discussion}\label{sec-discussion}

In many research areas first a primary study is analyzed, then a
follow-up study is analyzed with the goal to corroborate the
findings, or at least a subset of the findings, of the primary
study. We suggested novel testing procedures for corroborating the
evidence from a primary study in a follow-up study. We demonstrated
their usefulness on a GWAS application.
In the setting where there is no division of roles to a primary and
a follow-up study, the simulations suggested that our novel
Procedure \ref{procfdrsym} with $w_1 = 0.5$ is more powerful than
the BH procedure on maximum $p$-values.

 We proved that
Procedures \ref{procfdr} and \ref{procfdrsym} control the FDR when
the $p$-values are  independent within each study and the selection
rule is valid. However, the assumption of independence may not be realistic. Extensive simulations demonstrated that the BH procedure
controls the FDR for many types of dependence encountered in
practice \citep{yekutieli08b}. We conjecture that this robustness
property carries over to Procedures \ref{procfdr} and \ref{procfdrsym}, since Procedure \ref{procfdr} can be viewed as two-dimensional variant of the BH procedure. For simulated GWAS examples
the average false discovery proportion was below the nominal FDR
level, suggesting that the procedures are indeed valid for the type of
dependency that occurs in GWAS. More conservative variants of Procedure \ref{procfdr} were given in Theorem \ref{genthm} and in Section 3 of the Supplementary Material, that guarantee that the FDR is controlled for arbitrary dependence among the primary study $p$-values, and dependence of type PRDS  or arbitrary dependence  among the follow-up study $p$-values.
We demonstrated the usefulness of the  variants suggested in Theorem \ref{genthm} in Example 2 of Section \ref{sec-example}. Out of the 36 replicability discoveries with Procedure \ref{procfdr}, 23 discoveries passed the more stringent requirement that came with the added guarantee that the FDR is controlled for arbitrary dependence among the 635,547 $p$-values in the primary study.

Replicability analysis, as suggested in this paper, requires that
the investigators make several key design choices in addition to the
error level $q$: the selection rule, $q_1$, and $w_1$ if two studies
are available without division into primary and follow-up. The power
of the procedure for replicability analysis varies with these
choices. From our investigations,  it appears reasonable  in
Procedure \ref{procfdr} to select hypotheses by  BH at level $q_1$,
and to set $w_1 = 0.5$ in Procedure \ref{procfdrsym} if the
$p$-value distributions for false null hypotheses
may be assumed to be similar in both studies. We gave some
guidelines for choosing $q_1$ in specific settings, and more general
guidelines are a topic for future research.

In replicability analysis, the primary study guides the design of
the follow-up study by supplying the subset of hypotheses to be
followed-up. Since the primary study also yields information on
effect sizes, if it is assumed that the effect sizes are the
same across studies, then this information may be used in order to
determine the sample size needed to obtain good power in the
follow-up study. However, this assumption may be unrealistic in
applications such as GWAS, where the LD pattern varies across
populations.

Finally, we saw that although Procedure \ref{procfdrsym} with
parameters $(w_1, q_1,q)$ is far less conservative than the BH
procedure at level $q$ on maximum $p$-values, it is still
conservative. We proved that Procedure \ref{procfdrsym} with less
conservative parameters $q_1'>q_1$ and $q'>q$, still controls the
FDR at level $q$ on the family of no replicability null hypotheses,
if $|I_{00}|$ and $|I_{01}|$ were known. In future research we will
consider estimates of these unknown parameters.

\appendix

\section{Proof of Theorem \ref{thmfdr}}\label{proof-thmfdr}
 Let $q_2=q-q_1$, and for each
$j\in\{1,\ldots,m\},$ let  $P_1^{(j)}$ and $P_2^{(j)}$ denote the
vectors $P_1=(P_{11},  \ldots, P_{1m})$ and $P_2=(P_{21}, \ldots,
P_{2m})$ with, respectively, $P_{1j}$ and $P_{2j}$ excluded. For
$j\in\{1,\ldots,m\}$ arbitrary fixed, let
$\mathcal{R}_1^{(j)}(P_1^{(j)})\subseteq\{1,\ldots,j-1,
j+1,\ldots,m\}$ be the subset of indices selected along with index
$j$.  Note that since the selection rule is valid, this subset is
well defined. For any $j\in\{1,\ldots,m\}$ and given $P_1^{(j)}$, for $i\in{1,\ldots,j-1,j+1,\ldots,m}$ we define 
\begin{align*}
T_i = \left\{
\begin{array}{cl}
\max\left(\frac{mP_{1i}}{q_1},\,\frac{(|\mathcal{R}_1^{(j)}(P_1^{(j)})|+1)P_{2i}}{q_2}\right) & \text{if } i\in \mathcal{R}_1^{(j)}(P_1^{(j)}),\\
 \infty& \text{otherwise. } \\
\end{array} \right.
\end{align*}
Let $T_{(1)}\leq\ldots\leq T_{(m-1)}$ be the sorted $T$-values, and
$T_{(0)}=0$. For $r=1,\ldots,m$, we define $C_r^{(j)}$ as the event in
which if $H_{NR,j}$ is rejected by Procedure \ref{procfdr}, $r$
hypotheses are rejected including $H_{NR,j}$:
\begin{align*}
C_r^{(j)}=\{(P_1^{(j)}, P_2^{(j)}):\,T_{(r-1)}\leq r, T_{(r)}>r+1,
T_{(r+1)}>r+2,\ldots, T_{(m-1)}>m\}.
\end{align*}

Note that given $P_1$, for $r>|\mathcal{R}_1|$, $C_r^{(j)} = \emptyset$, since
exactly $|\mathcal{R}_1|-1$ $T_i$'s are finite.

 Obviously, $C_r^{(j)}$ and
$C_{r'}^{(j)}$ are disjoint events for any $r\neq r',$ and
$\cup_{r=1}^m C_r^{(j)}$ is the entire space of $(P_1^{(j)},
P_2^{(j)})$. Let $I_0 = I_{01}\cup I_{00}$,  $R_j$ be the indicator
of whether $H_{NR,j}$ was rejected for $j=1,\ldots,m$, and $R =
\sum_{j=1}^m R_j$. The FDR for the family of no replicability null
hypotheses is
\begin{align}
FDR=E\left(\frac{\sum_{j\in I_{0}}R_j}{\max(R,
1)}\right)+E\left(\frac{\sum_{j\in I_{10}}R_j}{\max(R,
1)}\right)\label{mainsum}
\end{align}

First, we find an upper bound for the first term of the sum in
(\ref{mainsum}).
\begin{align}
&E\left(\frac{\sum_{j\in I_{0}}R_j}{\max(R,
1)}\right)
=\sum_{j\in I_0}\sum_{r=1}^m\frac{1}{r}\textmd{Pr}\left(j\in
\mathcal{R}_1, P_{1j}\leq\frac{rq_1}{m},
P_{2j}\leq\frac{r(q-q_1)}{|\mathcal{R}_1|},
C_{r}^{(j)}\right)\notag\\&\leq\sum_{j\in
I_0}\sum_{r=1}^m\frac{1}{r}\textmd{Pr}\left(P_{1j}
\leq\frac{rq_1}{m}, C_{r}^{(j)}\right)=\sum_{j\in
I_0}\sum_{r=1}^m\frac{1}{r}\textmd{Pr}\left(P_{1j}
\leq\frac{rq_1}{m}\right)
\textmd{Pr}\left(C_{r}^{(j)}\right)\label{firstdep}\\&\leq
\frac{q_1}{m}\sum_{j\in
I_0}\sum_{r=1}^m\textmd{Pr}\left(C_{r}^{(j)}\right)
=\frac{|I_0|}{m}q_1\label{fullpart2}
\end{align}
The equality in (\ref{firstdep}) follows from the independence of
the $p$-values. The inequality in (\ref{fullpart2}) follows from the
fact that for each $j\in I_0$, $\textmd{Pr}(P_{1j}\leq x)\leq x$ for
all
$x\in[0, 1].$ 
Finally, the equality in (\ref{fullpart2}) follows from the fact
that $\cup_{r=1}^m C_r^{(j)}$ is the entire sample space of
$(P_1^{(j)}, P_2^{(j)}),$ represented as a union of disjoint events.

Next, we find an upper bound for the second term of the sum in
(\ref{mainsum}).  Let $\mathcal{R}_1(p_1)$ be the set of selected
indices using $P_1=p_1.$ Then $E\left(\sum_{j\in I_{10}}R_j/\max(R,
1)\,|\,P_1=p_1\right)$ equals to:
\begin{align}
&\sum_{j\in I_{10}\cap
\mathcal{R}_1(p_1)}\sum_{r=1}^{|\mathcal{R}_1(p_1)|}\frac{1}{r}\,\textbf{I}\left[p_{1j}\leq
\frac{rq_1}{m}\right]\textmd{Pr}\left(
P_{2j}\leq\frac{rq_2}{|\mathcal{R}_1(p_1)|},
C_{r}^{(j)}\,|\,P_1=p_1\right)\notag
\\&\leq\sum_{j\in I_{10}\cap
\mathcal{R}_1(p_1)}\sum_{r=1}^{|\mathcal{R}_1(p_1)|}\frac{1}{r}\textmd{Pr}\left(
P_{2j}\leq\frac{rq_2}{|\mathcal{R}_1(p_1)|},
C_{r}^{(j)}\,|\,P_1=p_1\right)\label{secdep}
\\&=\sum_{j\in I_{10}\cap\mathcal{R}_1(p_1)}\sum_{r=1}^{|\mathcal{R}_1(p_1)|}\frac{1}{r}
\textmd{Pr}\left(P_{2j}\leq\frac{rq_2}{|\mathcal{R}_1(p_1)|}\,|\,P_1=p_1\right)\textmd{Pr}\left(C_{r}^{(j)}\,|\,P_1=p_1\right)\label{indepsec}
\\&\leq \frac{q_2}{|\mathcal{R}_1(p_1)|}\sum_{j\in
I_{10}\cap\mathcal{R}_1(p_1)}\sum_{r=1}^{|\mathcal{R}_1(p_1)|}
\textmd{Pr}\left(C_{r}^{(j)}\,|\,P_1=p_1\right) =
\frac{q_2}{|\mathcal{R}_1(p_1)|}|I_{10}\cap\mathcal{R}_1(p_1)|.\label{fullsec}
\end{align}

The equality in (\ref{indepsec}) follows from the fact that $P_{2j}, P_2^{(j)},P_1$ are independent, since then
 $C_r^{(j)}$ and the event $\left\{P_{2j}\leq
rq_2/|\mathcal{R}_1(p_1)| \right\}$ are conditionally
independent. The inequality in
(\ref{fullsec}) follows from the independence of the $p$-values
across the studies and the fact that for each $j\in I_{10},$
$\textmd{Pr}(P_{2j}\leq x)\leq x$ for all $x\in[0, 1].$ The equality
in (\ref{fullsec}) follows from the fact that
$\cup_{r=1}^{|\mathcal{R}_1(p_1)|}C_{r}^{(j)}$ is a union of
disjoint events, and
$\textmd{Pr}\left(\cup_{r=1}^{|\mathcal{R}_1(p_1)|}C_{r}^{(j)}\,|\,P_1=p_1\right)=1.$

It follows from (\ref{fullsec}) that $E\left(\sum_{j\in
I_{10}}R_j/\max(R, 1)\right)\leq q_2$.  Using this fact and the
bound (\ref{fullpart2}) for the first term of (\ref{mainsum}), we
obtain:
\begin{align*}
FDR\leq\frac{|I_0|}{m}q_1+(q-q_1)\leq q_1+(q-q_1)=q.
\end{align*}
\section{Proof for FDR control of the oracle Procedure \ref{thmfdr}}\label{proof-thmfdr-oracle}
Let us now prove that under the assumption that the $p$-values are
independent, Procedure \ref{procfdr} at levels $\left(q',
2q'\right)$ controls the FDR  at level
$|I_{00}|\left(q'\right)^2/m+\left(|I_{01}|/m+1\right)q'$. Returning
to the proof of Theorem \ref{thmfdr}, note that (\ref{mainsum}) can
be rewritten as follows.
\begin{align}
FDR=E\left(\frac{\sum_{j\in I_{00}}R_j}{\max(R,
1)}\right)+E\left(\frac{\sum_{j\in I_{01}}R_j}{\max(R,
1)}\right)+E\left(\frac{\sum_{j\in I_{10}}R_j}{\max(R,
1)}\right).\label{mainsum2}
\end{align}
We will now give an upper bound for each term of the sum in
(\ref{mainsum2}). First,
\begin{align}
&E\left(\frac{\sum_{j\in I_{00}}R_j}{\max(R, 1)}\right)=\sum_{j\in
I_{00}}\sum_{r=1}^m\frac{1}{r}\textmd{Pr}\left(j\in \mathcal{R}_1,
P_{1j}\leq\frac{rq'}{m}, P_{2j}\leq\frac{rq'}{|\mathcal{R}_1|},
C_{r}^{(j)}\right)\notag\\&\leq\sum_{j\in
I_{00}}\sum_{r=1}^m\frac{1}{r}\textmd{Pr}\left(P_{1j}
\leq\frac{rq'}{m}, P_{2j}\leq q', C_{r}^{(j)}\right)
\leq\frac{(q')^2}{m}\sum_{j\in
I_{00}}\sum_{r=1}^m\textmd{Pr}\left(C_{r}^{(j)}\right)
=\frac{|I_{00}|}{m}(q')^2\label{full3}
\end{align}
The second inequality in (\ref{full3}) follows from the facts that
for each $j\in I_{00}$, $P_{1j}$ and $P_{2j}$ are independent, and
$\textmd{Pr}(P_{ij}\leq x)\leq x$ for all $x\in[0,1]$ and $i=1,2.$
The equality in (\ref{full3}) follows from the explanation of the equality in (\ref{fullpart2}).

Second, replacing $I_{0}$ by $I_{01}$ and $|I_0|$ by $|I_{01}|$ in
the arguments that led to (\ref{fullpart2}), we obtain:
\begin{align}
E\left(\frac{\sum_{j\in I_{01}}R_j}{\max(R, 1)}\right)\leq
\frac{|I_{01}|}{m}q'.\label{secterm}
\end{align}

Finally, using (\ref{fullsec}) in the proof of Theorem \ref{thmfdr}
we obtain that the third term of the sum in (\ref{mainsum2}) is
bounded by $q_2 =2q'-q'=q'$. Using this upper bound, together with
the bounds for the first two terms derived in (\ref{full3}) and
(\ref{secterm}), we obtain:
\begin{align*}
FDR\leq\frac{|I_{00}|}{m}(q')^2+\frac{|I_{01}|}{m}q'+q'=\frac{|I_{00}|}{m}(q')^2+\left(\frac{|I_{01}|}{m}+1\right)q'.
\end{align*}
It follows that if $|I_{00}|$ and $|I_{01}|$ were known, one could
guarantee FDR control at level $q$ on the family of no replicability
null hypotheses by applying Procedure \ref{procfdr} at levels $(q',
2q')$, where $q'$ is the solution to
$|I_{00}|\left(q'\right)^2/m+\left(|I_{01}|/m+1\right)q'=q.$

\section{Proof of Theorem \ref{thmprocfdrsym}}\label{proof-thmprocfdrsym}
Let $V_{12}=\sum_{j\in I_{00}\cup I_{01}\cup
I_{10}}\textbf{I}\left[j\in\mathcal R_{12, w_1q}\right]$ and
$R_{12}=|\mathcal R_{12, w_1q}|$ denote the number of erroneously
rejected and the total number of rejected no replicability null
hypotheses by Procedure \ref{procfdr} at level $w_1q$ with study one
as the primary study and study two as the follow-up study.
Similarly, let $V_{21}=\sum_{j\in I_{00}\cup I_{01}\cup
I_{10}}\textbf{I}\left[j\in \mathcal R_{21, (1-w_1)q}\right]$ and
$R_{21}=|\mathcal R_{21, (1-w_1)q}|$ denote the number of
erroneously rejected and the total number of rejected no
replicability null hypotheses by Procedure \ref{procfdr} at level
$(1-w_1)q$ with study two as the primary study and study one as the
follow-up study. Define $\mathcal R_{s}=\mathcal R_{12,
w_1q}\cup\mathcal R_{21, (1-w_1)q}$, the indices of the no
replicability null hypotheses rejected by Procedure
\ref{procfdrsym}. Let $V_{s}=\sum_{j\in I_{00}\cup I_{01}\cup
I_{10}}\textbf{I}\left[j\in\mathcal R_{s}\right]$ and
$R_{s}=|\mathcal R_{s}|,$ the number of erroneously rejected and the
total number of rejected
no replicability null hypotheses by Procedure \ref{procfdrsym}. 

Note that $V_s\leq V_{12}+V_{21}$. Therefore,
\begin{align}
FDR=E\left(\frac{V_s}{\max(R_s, 1)}\right)\leq
E\left(\frac{V_{12}}{\max(R_s,
1)}\right)+E\left(\frac{V_{21}}{\max(R_s, 1)}\right).\label{lastin}
\end{align}
In addition, note that
 $\max (R_s, 1)\geq
\max(R_{12}, 1)$ and $\max(R_s, 1)\geq \max(R_{21}, 1).$ Using these
facts and (\ref{lastin}) we obtain
\begin{align*}
FDR=E\left(\frac{V_s}{\max(R_s, 1)}\right)\leq
E\left(\frac{V_{12}}{\max(R_{12},
1)}\right)+E\left(\frac{V_{21}}{\max(R_{21}, 1)}\right)\leq
w_1q+(1-w_1)q=q,
\end{align*}
where the last inequality follows from Theorem \ref{thmfdr}.

\section{Table of results for GWAS of Crohn's disease }\label{sec-CrohnTable}

\begin{table}
\caption{Replicability analysis for Example 2 in Section \ref{sec-example}: GWAS of Crohn’s disease. The number of SNPs in the primary study was
635,547, and 126 SNPs were followed-up. The 36 discoveries by Procedure \ref{procfdr} with parameters $(q_1,q) = (0.04,0.05)$ are listed according to the adjusted $p$-values.
The primary and follow-up studies $p$-values are given in columns 4 and 5;  the adjusted $p$-values for $c=0.8$ are given in column 6 for Procedure \ref{procfdr}, and in column 7 for the modification of item 1 in Theorem 3.3. }\label{tab-crohn}
\centering
\begin{tabular}{rllllll}
  \hline
Index & Chromosome & Position & $p_1$ & $p_2$ & $p^{REPadj}_{FDR}$ &$\tilde{p}^{REPadj}_{FDR}$   \\
  \hline
1 & 1 & 67417979 & 3.19e-34 & 1.5e-36 & 2.53e-28 & 3.53e-27 \\
  2 & 1 & 67414547 & 5.05e-36 & 3.1e-29 & 9.69e-27 & 9.69e-27 \\
  3 & 1 & 67387537 & 1.35e-24 & 5.62e-17 & 1.17e-14 & 1.17e-14 \\
  4 & 2 & 233962410 & 5.66e-21 & 7.67e-14 & 1.2e-11 & 1.2e-11 \\
  5 & 10 & 64108492 & 9.51e-12 & 1.61e-10 & 1.51e-06 & 1.5e-05 \\
  6 & 5 & 40428485 & 2.51e-22 & 2.79e-08 & 2.84e-06 & 3.31e-06 \\
  7 & 5 & 40437266 & 2.26e-22 & 3.18e-08 & 2.84e-06 & 3.31e-06 \\
  8 & 10 & 101281583 & 8.53e-11 & 1.69e-07 & 1.32e-05 & 7.74e-05 \\
  9 & 18 & 12769947 & 5.95e-12 & 2.41e-07 & 1.61e-05 & 1.88e-05 \\
  10 & 5 & 150239060 & 3.18e-11 & 2.57e-07 & 1.61e-05 & 3.91e-05 \\
  11 & 10 & 101282445 & 9.09e-11 & 3.1e-07 & 1.76e-05 & 7.74e-05 \\
  12 & 5 & 150203580 & 4.09e-11 & 7.47e-07 & 3.89e-05 & 4.67e-05 \\
  13 & 18 & 12799340 & 3.27e-11 & 1.23e-06 & 5.91e-05 & 6.99e-05 \\
  14 & 5 & 131798704 & 2.29e-09 & 3.52e-11 & 0.00013 & 0.00169 \\
  15 & 5 & 158747111 & 4.4e-09 & 3.66e-06 & 0.000233 & 0.00305 \\
  16 & 2 & 233965368 & 1.28e-21 & 3.66e-05 & 0.00143 & 0.00163 \\
  17 & 13 & 43355925 & 8.04e-08 & 1.33e-07 & 0.00376 & 0.0469 \\
  18 & 12 & 39104262 & 8.95e-08 & 6.55e-05 & 0.00395 & 0.0496 \\
  19 & 3 & 49676987 & 9.47e-08 & 2.24e-06 & 0.00396 & 0.0499 \\
  20 & 3 & 49696536 & 1.08e-07 & 5.64e-07 & 0.00429 & 0.0544 \\
  21 & 12 & 38888207 & 6.64e-08 & 0.000165 & 0.00491 & 0.0433 \\
  22 & 6 & 167408399 & 1.65e-07 & 3.26e-07 & 0.00596 & 0.0731 \\
  23 & 9 & 114645994 & 1.96e-07 & 6.58e-05 & 0.00677 & 0.0768 \\
  24 & 6 & 20836710 & 1.26e-07 & 0.000278 & 0.00724 & 0.0607 \\
  25 & 1 & 169593891 & 2.01e-07 & 0.000321 & 0.00802 & 0.0768 \\
  26 & 1 & 197667523 & 3.41e-07 & 2.34e-06 & 0.01 & 0.111 \\
  27 & 9 & 4971602 & 3.4e-07 & 0.00043 & 0.01 & 0.111 \\
  28 & 1 & 157665119 & 1.75e-07 & 0.000481 & 0.0107 & 0.0745 \\
  29 & 11 & 75978964 & 7.16e-08 & 0.000732 & 0.0158 & 0.044 \\
  30 & 20 & 61798026 & 7.6e-07 & 0.000138 & 0.0201 & 0.234 \\
  31 & 6 & 167405736 & 1.65e-07 & 0.00121 & 0.0241 & 0.0731 \\
  32 & 1 & 197691964 & 9.69e-07 & 1e-04 & 0.0241 & 0.29 \\
  33 & 17 & 35294289 & 1.06e-06 & 0.000292 & 0.0255 & 0.308 \\
  34 & 8 & 126603853 & 1.9e-06 & 0.000182 & 0.0431 & 0.457 \\
  35 & 6 & 106541962 & 1.85e-06 & 7.7e-06 & 0.0431 & 0.457 \\
  36 & 9 & 4978761 & 1.96e-06 & 0.00162 & 0.0433 & 0.462 \\
   \hline
\end{tabular}
\end{table}
\newpage
\title{ Supplementary Material for Discovering associations that replicate from a primary study of high dimension to a follow-up study }

\maketitle

\begin{center}

Marina Bogomolov \\
\emph{Faculty of Industrial  Engineering and Management, Technion --
Israel Institute of Technology, Haifa, Israel. E-mail:
marinabo@tx.technion.ac.il }\\
Ruth Heller \\
\emph{Department of Statistics and Operations Research, Tel-Aviv
university, Tel-Aviv, Israel. E-mail: ruheller@post.tau.ac.il}\\
\end{center}

\section{A computational example with FWER
control}

When the FWER controlling procedure applied in each stage of
Procedure 3.1 is Bonferroni, then $H_{NR,j}$ is rejected if
$p_{1j}\leq \alpha_1/m$ and $p_{2j}\leq
(\alpha-\alpha_1)/\sum_{i=1}^m \textbf{I}[p_{1j}\leq \alpha_1/m]$,
where $\textbf{I}[\cdot]$ is the indicator function.
 An alternative to Procedure 3.1 is to apply a FWER
controlling procedure, such as Bonferroni, on the maximum of
$p$-values from the two studies. This alternative procedure also
controls the FWER on the family of no replicability null hypotheses.
In the alternative procedure, $H_j$ is rejected if $p_{1j}\leq
\alpha/m$ and $p_{2j}\leq \alpha/m$. The two procedures differ in
the thresholds used in each of the studies.
 The cut-off for $p_{1j}$ is
larger in the alternative procedure, since $\alpha_1< \alpha$.
However, the cut-off for $p_{2j}$ may be substantially smaller in
the alternative procedure, since $(\alpha-\alpha_1)/\sum_{i=1}^m
\textbf{I}[p_{1j}\leq \alpha_1/m]$ may be significantly larger than
$\alpha/m$. This is so in the common setting where signal is sparse
in the primary study, i.e. $\sum_{j=1}^m h_{1j}\ll m$.

\begin{ex}\label{ex1}
 Suppose
we have $m$ independent normal outcomes in each of the two studies
$T_{1j}, T_{2j}, j=1\ldots,m$ . In this example, $E(T_{11}) =
\mu_{11}, E(T_{21}) = \mu_{21}, Var(T_{11}) = Var(T_{21})=1$, and
outcomes $j=2,\ldots,m$ have expectation 0 and variance 1. Consider
first the power of the alternative procedure that applies Bonferroni
on the maximum of the two study $p$-values for FWER control at level
$\alpha=0.05$:
$$\pi_1 = \stackrel \sim \Phi (z_{1-\alpha/m}-\mu_{11})\times \stackrel \sim \Phi (z_{1-\alpha/m}-\mu_{21}),
$$ where $\stackrel \sim \Phi(\cdot)$ is the right tail of the standard normal distribution.
Next, we compute the power of Procedure 3.1 with
 Bonferroni as the FWER
controlling procedure. The probability of correctly selecting (PCS)
the non-null hypothesis in the first study as well as $k-1$ null
hypotheses along with it is $$PCS(k) = \stackrel \sim \Phi
(z_{1-\alpha_1/m}-\mu_{11})\binom{m-1}{k-1}(\alpha_1/m)^{k-1}(1-\alpha_1/m)^{m-k},$$
so the power is $$\pi_2 = \sum_{k=1}^m PCS(k) \times \stackrel \sim
\Phi (z_{1-(\alpha-\alpha_1)/k}-\mu_{21}).$$

 Figure \ref{fig1} shows the
power of the Bonferroni on maximum $p$-values procedure (left panel)
and the power of Procedure 3.1 (right panel) for different
configurations of $(\mu_{11}, \mu_{21})$, where $(\alpha_1, \alpha)
= (0.025, 0.05)$. In most configurations of $\mu_{11}$ and
$\mu_{21}$, Procedure 3.1 is more powerful than the Bonferroni on
maximum $p$-values procedure. Moreover, for fixed $\mu_1>\mu_2$, the
power of the two stage procedure is larger if $(\mu_{11}, \mu_{21})
= (\mu_1, \mu_2)$ than if $(\mu_{11}, \mu_{21}) = (\mu_2, \mu_1)$.

Figure \ref{fig2} shows the difference in power of Procedure 3.1
using Bonferroni with $c=\alpha_1/\alpha \in \{0.2,0.5,0.8\}$, as
well as the Bonferroni procedure on maximum $p$-values, from the
power of Procedure 3.1 with optimal choice of $c$. Clearly,
Procedure 3.1 with optimal choice of $c$ can be much more powerful
than the Bonferroni procedure on maximum $p$-values. Moreover, for
the three choices $c=0.2$, $c=0.5$ and $c=0.8$, the difference in
power from the optimal power is fairly small, especially when the
optimal power is above 0.9 (right panel). Figure \ref{fig3} shows
the power as a function of $c$ for three configurations of
$(\mu_{11}, \mu_{21})$, for which the power using the optimal $c$ is
0.9. The power function is quite flat. The optimal $c$ is below 0.5
in the top left panel, and above 0.5 in the top right and bottom
panel. However, the difference in power between Procedure 3.1 with
$c=0.5$ and Procedure 3.1 with optimal $c$ is small.
\begin{figure}[!tpb]
\centering
\includegraphics[width=6cm, height=6cm]{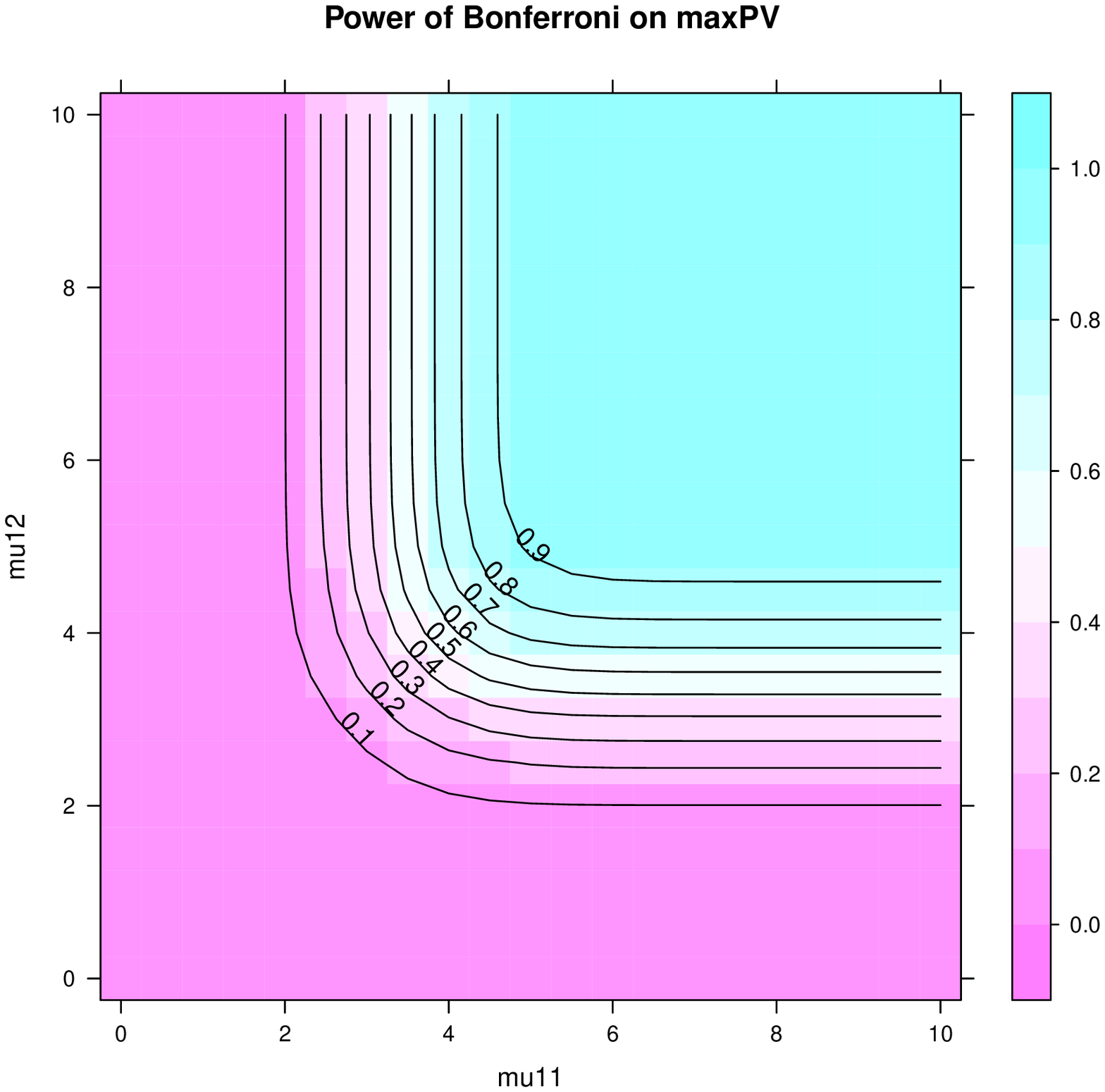}
\includegraphics[width=6cm, height=6cm]{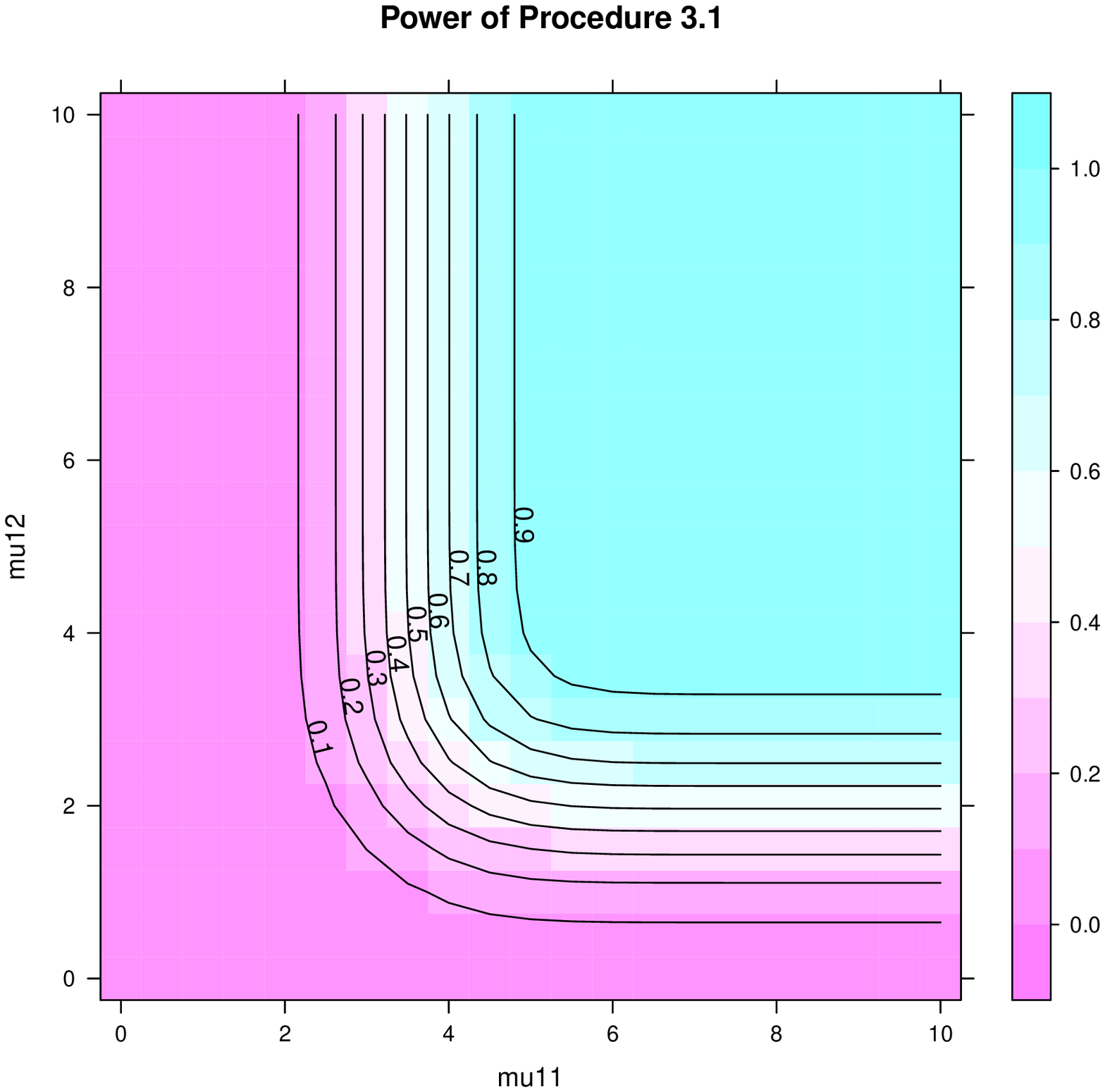}
\caption{ The power as function of the expectation in the first
study (x-axis) and the expectation in the second study (y-axis), for
the false no replicability null hypothesis, in a setting where one
no replicability null hypothesis is false out of 100  no
replicability null hypotheses. Left panel: Procedure that applies a
Bonferroni correction on the maximum two study $p$-values for FWER
control at level 0.05. Right panel:  Procedure 3.1 with
$(\alpha_1,\alpha) =( 0.025, 0.05)$ and Bonferroni as the FWER
controlling procedure. }\label{fig1}
\end{figure}

\begin{figure}[!tpb]
\centering
\includegraphics[width=7cm, height=7cm]{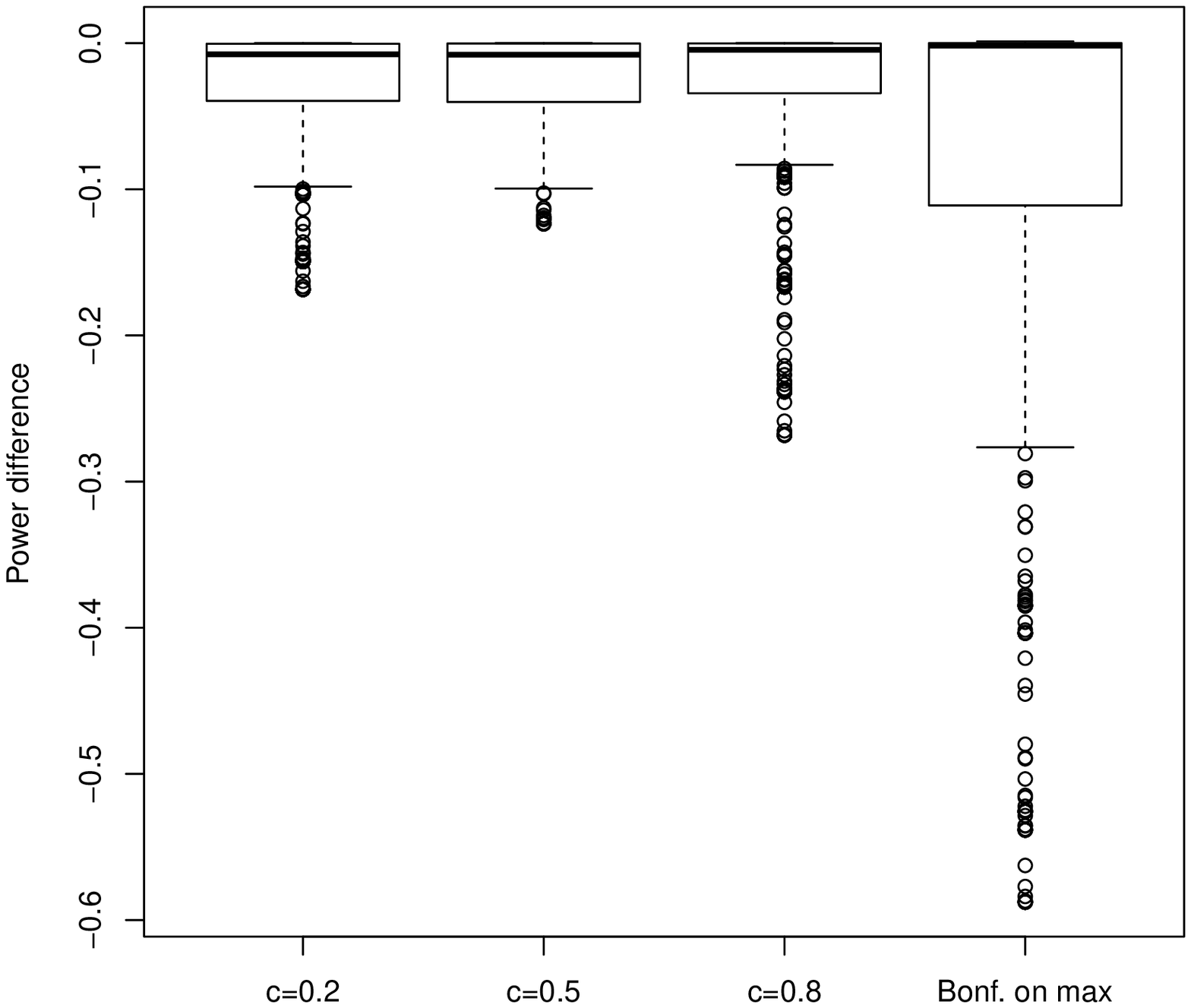}
\includegraphics[width=7cm, height=7cm]{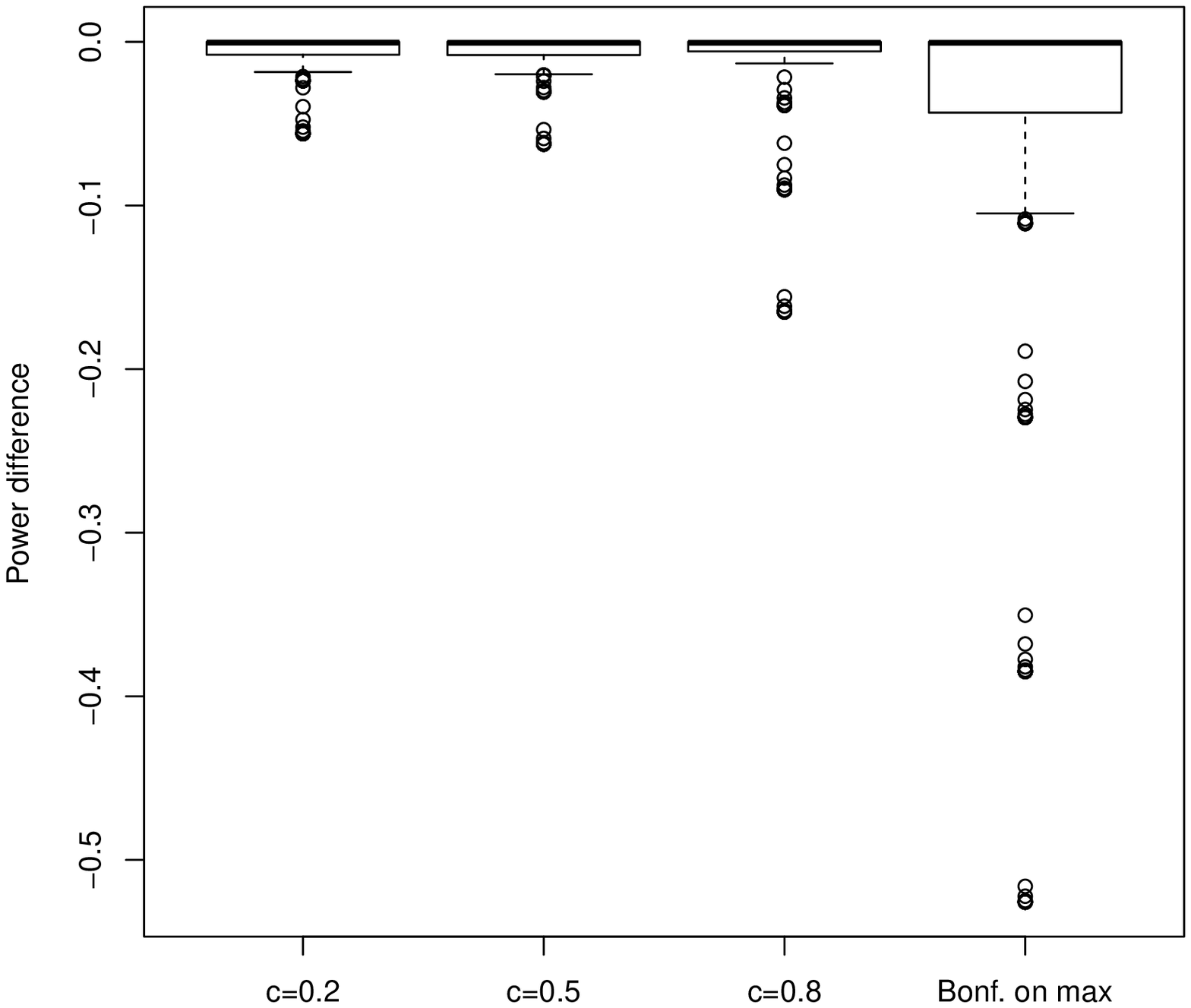}
\caption{ The difference in power of Procedure 3.1 using Bonferroni
with $c=\alpha_1/\alpha \in \{0.2,0.5,0.8\}$, as well as the
Bonferroni procedure on maximum $p$-values, from the power of
Procedure 3.1 with optimal $c$. Left panel: for all pairs of
configurations where $\mu_{11}\in \{0, 0.5, 1.0,\ldots,10 \}$ and
$\mu_{21}\in \{0, 0.5, 1.0,\ldots,10 \}$. Right panel: Subset of
configurations of $(\mu_{11},\mu_{21})$ for which the power with
optimal choice $c$ is above 0.90.}\label{fig2}
\end{figure}

\begin{figure}[!tpb]
\centering
\includegraphics[width=6cm, height=6cm]{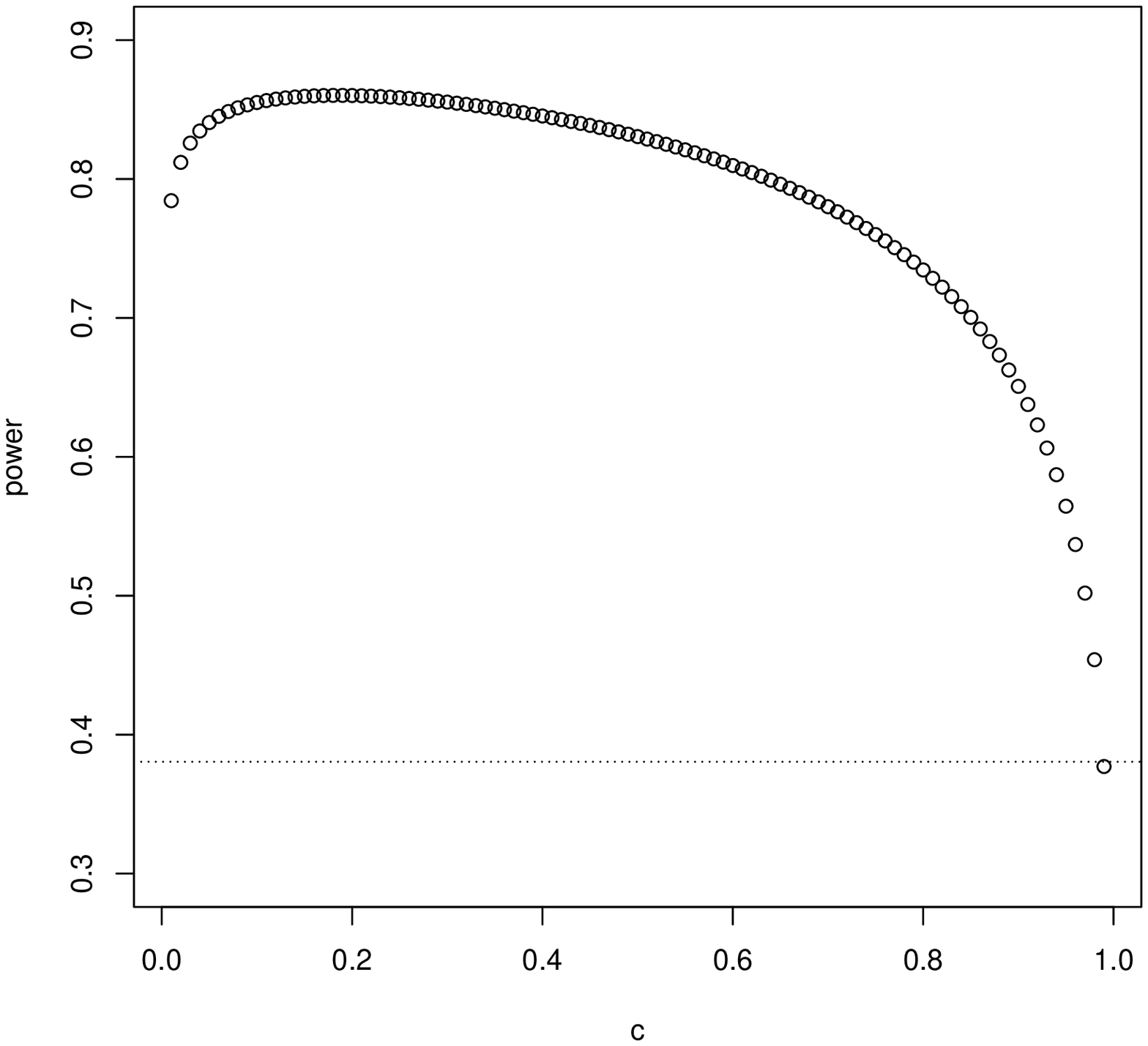}
\includegraphics[width=6cm, height=6cm]{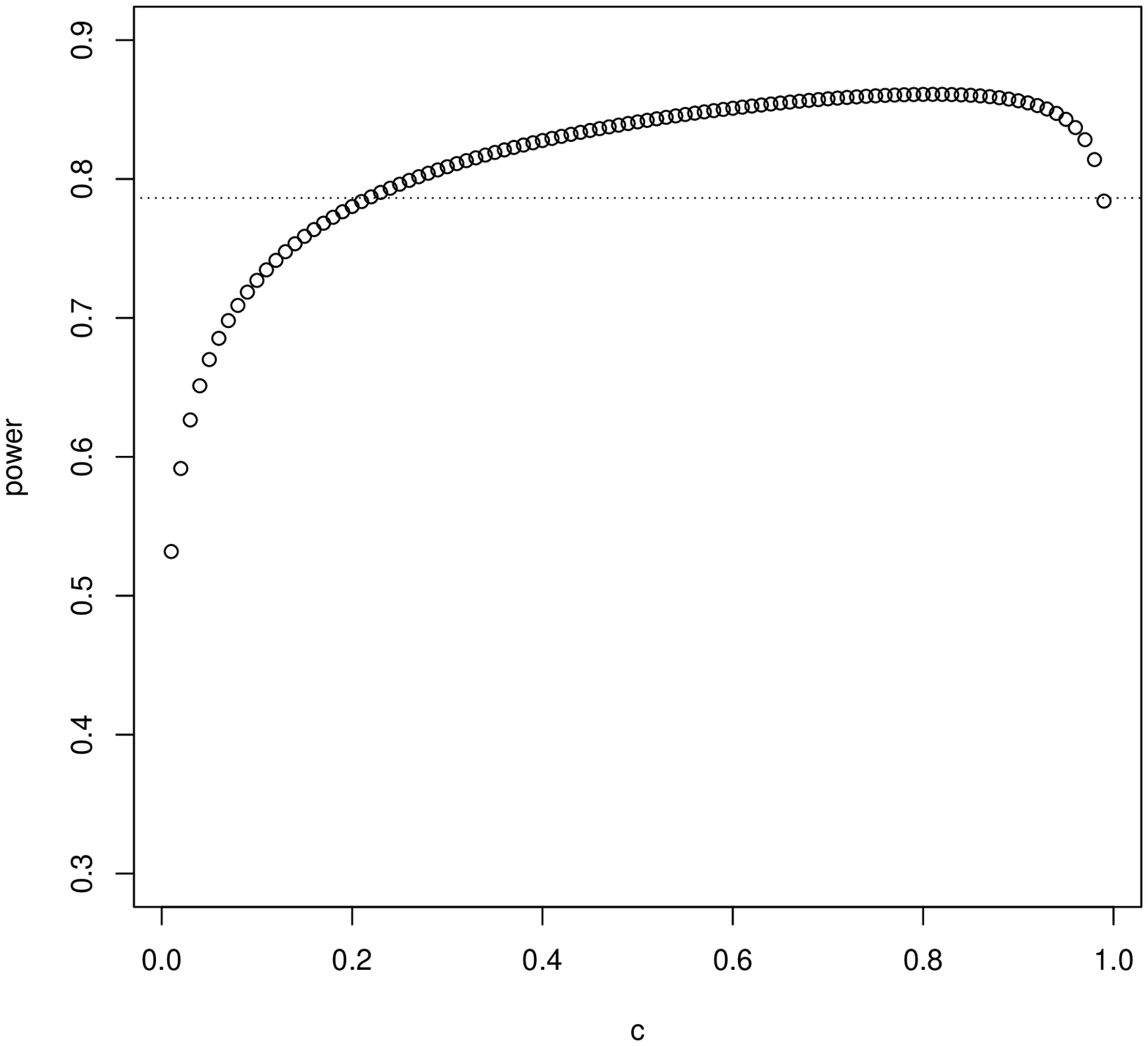}
\includegraphics[width=6cm, height=6cm]{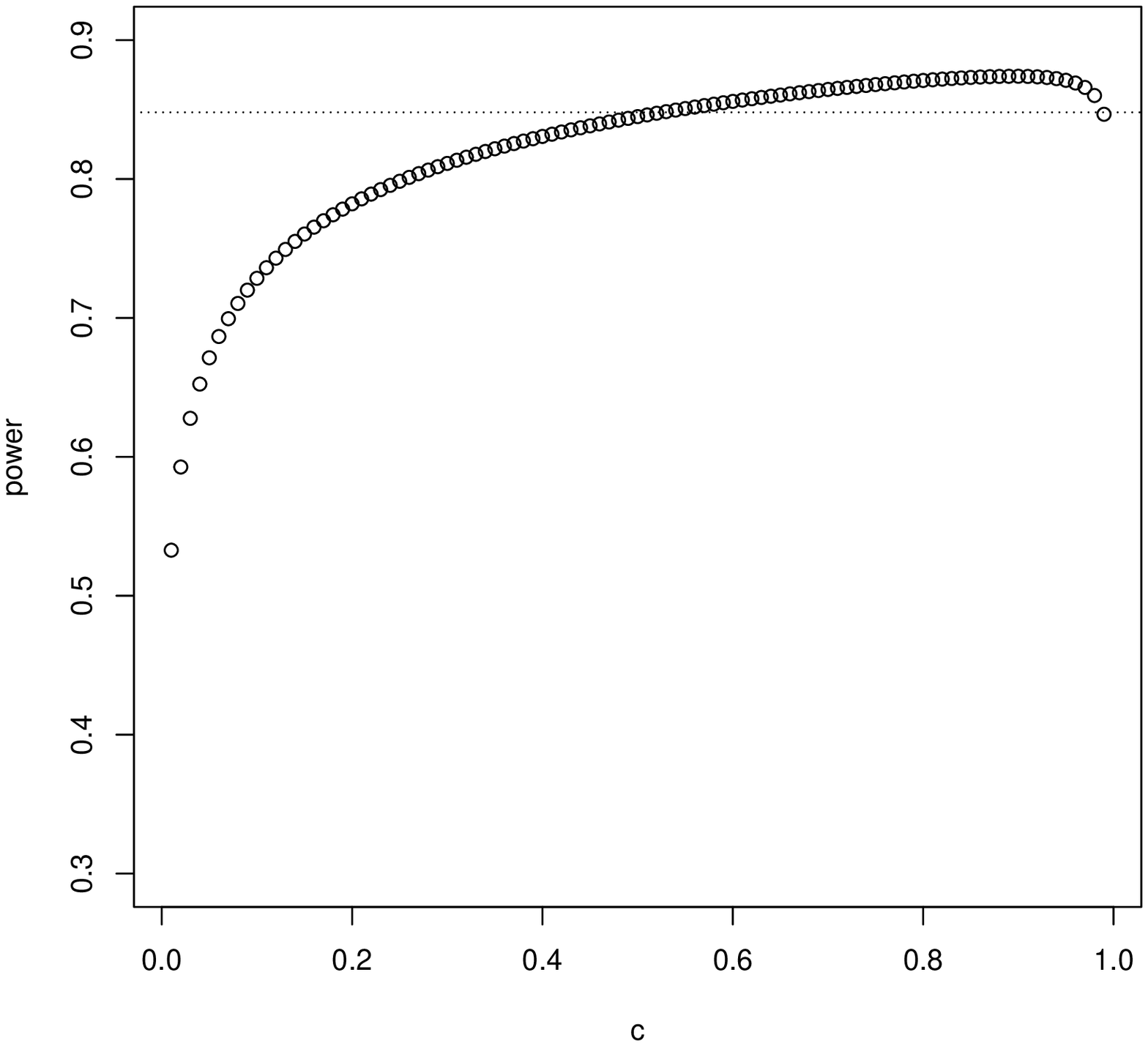}
\caption{ The power of Procedure 3.1 using Bonferroni as function of
$c = \alpha_1/\alpha$ for the false no replicability null
hypothesis, for the following configurations of $(\mu_{11},
\mu_{21})$: (5.5, 3.0) in the top left panel; (4.5, 4.5) in the top
right panel; (4.5, 5.0) in the bottom panel. The power of the
Bonferroni procedure on maximum $p$-values is the dotted horizontal
line. }\label{fig3}
\end{figure}

\end{ex}

\section{Proof of Theorem 3.3}
We use the notation given in the first two paragraphs of Appendix A of
the main manuscript, including: $q_2=q-q_1;$ $R_j$ is the indicator
of whether $H_{NR,j}$ was rejected for $j=1,\ldots,m$, and $R =
\sum_{j=1}^m R_j$. In addition we define:  $I_0=I_{00}\cup
I_{01};$ $p_1=(p_{11},\ldots, p_{1m});$ $\mathcal{R}_1(p_1)$ is the
set of hypotheses selected for follow-up based on $p_1,$
$R_1(p_1)=|\mathcal{R}_1(p_1)|$.
\begin{lemma}\label{lemdepSM}
Assume that the $p$-values across studies are independent, and the
set of $p$-values within the follow-up study has property PRDS. Then
for any valid selection rule, the following results hold:
\begin{enumerate}
\item Given $p_1,$ for $j\in I_{10}\cap \mathcal{R}_1(p_1),$  $$\sum_{r=1}^{R_1(p_1)}\textmd{Pr}
\left(C_{r}^{(j)}\,|\,P_{2j}\leq\frac{rq_2}{R_1(p_1)},
P_1=p_1\right)\leq 1.$$
\item For Procedure 3.2 with parameters $(q_1,q)$,
$$E\left(\frac{\sum_{j\in I_{10}}R_j}{\max(R,
1)}\right)\leq q_2.$$
\item Item 2 holds if in the terms $rq_1/m$ and $R_2q_1/m$ in step 2 of Procedure 3.2, $q_1$ is replaced by $q_1',$ for any value of $q_1'.$
\end{enumerate}
\end{lemma}
See Section \ref{subsec-prooflemmaSM2} for a proof.


\textbf{Proof of item 1 of Theorem 3.3.} We will first show that the
first term of the sum in (A.1) is bounded by $|I_0|\,q_1/m.$ We will
use the technique developed in \cite{yoav6} in the proof of their
Theorem 1.3. For each $j\in I_0, r\in\{1,\ldots,m\},$ and
$l\in\{1,\ldots,m\},$ let us define:
\begin{align*}p_{jrl}=\textmd{Pr}\left(P_{1j}\in\left(\frac{(l-1)q_1}{m\sum_{s=1}^m\frac{1}{s}},
\frac{lq_1}{m\sum_{s=1}^m\frac{1}{s}} \right],
C_r^{(j)}\right).\end{align*} Since $\cup_{r=1}^m C_r^{(j)}$ is the
entire sample space represented as a union of disjoint events, we
obtain for each $j\in I_0$ and $l\in\{1,\ldots,m\}$:
\begin{align}
\sum_{r=1}^m p_{jrl}&=\textmd{Pr}\left(P_{1j}\in
\left(\frac{(l-1)q_1}{m\sum_{s=1}^m\frac{1}{s}},
\frac{lq_1}{m\sum_{s=1}^m\frac{1}{s}} \right],\,
\cup_{r=1}^mC_r^{(j)}\right)\notag\\&=\textmd{Pr}\left(P_{1j}\in
\left(\frac{(l-1)q_1}{m\sum_{s=1}^m\frac{1}{s}},
\frac{lq_1}{m\sum_{s=1}^m\frac{1}{s}}
\right]\right).\label{daniwhole}
\end{align}

Note that for $j\in I_0,$ $\textmd{Pr}\left(P_{1j}\leq x\right)\leq
x$ for all $x\geq 0$, in particular
$\textmd{Pr}\left(P_{1j}=0\right)=0.$ Therefore, for each $j\in I_0$
and $r\in\{1,\ldots,m\},$
\begin{align}
\textmd{Pr}\left(P_{1j} \leq\frac{rq_1}{m\sum_{s=1}^m\frac{1}{s}},
C_{r}^{(j)}\right)=\sum_{l=1}^{r}p_{jrl}.\label{sum1}
\end{align}

The upper bound on the first term of the sum in (A.1) is derived as follows.
\begin{align}
E\left(\frac{\sum_{j\in I_{0}}R_j}{\max(R, 1)}\right)&=\sum_{j\in
I_0}\sum_{r=1}^m\frac{1}{r}\textmd{Pr}\left(j\in \mathcal{R}_1,
P_{1j}\leq\frac{rq_1}{m\sum_{s=1}^m\frac{1}{s}},
P_{2j}\leq\frac{r(q-q_1)}{|\mathcal{R}_1|},
C_{r}^{(j)}\right)\notag
\notag\\&\leq\sum_{j\in
I_0}\sum_{r=1}^m\frac{1}{r}\textmd{Pr}\left(P_{1j}
\leq\frac{rq_1}{m\sum_{s=1}^m\frac{1}{s}},
C_{r}^{(j)}\right)\label{foritem2}\\&=\sum_{j\in
I_0}\sum_{r=1}^m\sum_{l=1}^{r}\frac{1}{r}p_{jrl}=\sum_{j\in
I_0}\sum_{l=1}^m\sum_{r=l}^{m}\frac{1}{r}p_{jrl}\label{1eq}\\&\leq
\sum_{j\in
I_0}\sum_{l=1}^m\sum_{r=l}^{m}\frac{1}{l}p_{jrl}\leq\sum_{j\in
I_0}\sum_{l=1}^m\frac{1}{l}\sum_{r=1}^{m}p_{jrl}\notag\\&=\sum_{j\in
I_0}\sum_{l=1}^m\frac{1}{l}\textmd{Pr}\left(P_{1j}\in\left(\frac{(l-1)q_1}{m\sum_{s=1}^m\frac{1}{s}},
\frac{lq_1}{m\sum_{s=1}^m\frac{1}{s}}\right]\right),\label{3}
\end{align}
where the first equality in (\ref{1eq}) follows from (\ref{sum1}),
and the equality in (\ref{3}) follows from (\ref{daniwhole}). Note
that for each $j\in I_0,$
\begin{align}
\sum_{l=1}^m&\frac{1}{l}\textmd{Pr}\left(P_{1j}\in\left(\frac{(l-1)q_1}{m\sum_{s=1}^m\frac{1}{s}},
\frac{lq_1}{m\sum_{s=1}^m\frac{1}{s}}\right]\right)\notag\\&=
\sum_{l=1}^m\frac{1}{l}\left[\textmd{Pr}\left(P_{1j}\leq
\frac{lq_1}{m\sum_{s=1}^m1/s}\right)-\textmd{Pr}\left(P_{1j}\leq
\frac{(l-1)q_1}{m\sum_{s=1}^m1/s}\right)\right]\notag\\&=
\sum_{l=1}^m\frac{1}{l}\textmd{Pr}\left(P_{1j}\leq
\frac{lq_1}{m\sum_{s=1}^m1/s}\right)-\sum_{l=0}^{m-1}\frac{1}{l+1}\textmd{Pr}\left(P_{1j}\leq
\frac{lq_1}{m\sum_{s=1}^m1/s}\right)\notag
\\&=\sum_{l=1}^{m-1}\left(\frac{1}{l}-\frac{1}{l+1}
\right)\textmd{Pr}\left(P_{1j}\leq\frac{lq_1}{m\sum_{s=1}^m1/s}\right)+\frac{1}{m}\textmd{Pr}\left(P_{1j}\leq
\frac{q_1}{\sum_{s=1}^m1/s}\right)\notag
\\&\leq\sum_{l=1}^{m-1}\frac{1}{l+1}\left(\frac{q_1}{m\sum_{s=1}^m1/s}\right)+\frac{q_1}{m\sum_{s=1}^m1/s}\label{daniprelast}\\&=\left(\frac{q_1}{m\sum_{s=1}^m1/s}\right)\sum_{l=1}^{m}\frac{1}{l}=\frac{q_1}{m}.\label{danilast}
\end{align}
The  inequality in (\ref{daniprelast}) follows from the fact that
for $j\in I_0,$ $\textmd{Pr}(P_{1j}\leq x)\leq x$ for all $x\geq 0.$
Combining (\ref{danilast}) with (\ref{3}) we obtain an upper bound
for the first term of the sum in (A.1):
\begin{align}E\left(\frac{\sum_{j\in I_{0}}R_j}{\max(R,
1)}\right)\leq\sum_{j\in
I_0}\frac{q_1}{m}=\frac{|I_0|\,q_1}{m}.\label{item1}\end{align}

It follows from Lemma \ref{lemdepSM}, item 3, that the second term
of the sum in (A.1) is bounded by $q_2,$ hence
$$FDR\leq \frac{|I_0|\,q_1}{m}+q_2=\frac{|I_0|\,q_1}{m}+q-q_1\leq q.$$

\textbf{Proof of item 2 of Theorem 3.3.} We will first prove that
the first term of the sum in (A.1) is bounded by $q_1.$  For
$\widetilde{q}_1$ as defined in item 2 of Theorem 3.3, we denote
$k=\left\lceil tm/\widetilde{q}_1-1\right\rceil.$ The first term of
the sum in (A.1) is upper bounded by two terms:
\begin{align}
E\left(\frac{\sum_{j\in I_{0}}R_j}{\max(R, 1)}\right)&=\sum_{j\in
I_0}\sum_{r=1}^m\frac{1}{r}\textmd{Pr}\left(j\in \mathcal{R}_1,
P_{1j}\leq\frac{r\widetilde{q_1}}{m},
P_{2j}\leq\frac{r(q-q_1)}{|\mathcal{R}_1|},
C_{r}^{(j)}\right)\notag
\notag\\&\leq\sum_{j\in
I_0}\sum_{r=1}^m\frac{1}{r}\textmd{Pr}\left(P_{1j}
\leq\min\left(\frac{r\widetilde{q_1}}{m},\,t\right),\,
C_{r}^{(j)}\right)\label{min}\\&= \sum_{j\in
I_0}\sum_{r=1}^k\frac{1}{r}\textmd{Pr}\left(P_{1j}
\leq\frac{r\widetilde{q_1}}{m},\, C_{r}^{(j)}\right)+\sum_{j\in
I_0}\sum_{r=k+1}^m\frac{1}{r}\textmd{Pr}\left(P_{1j} \leq t,\,
C_{r}^{(j)}\right),\label{sumitem3}
\end{align}
where the inequality in (\ref{min}) follows from the fact that $j\in
\mathcal{R}_1$ yields that $P_{1j}\leq t.$ We will now find an upper
bound for each of the two terms in (\ref{sumitem3}) separately. The
derivation of the upper bound for the first term  is along
the lines of the derivation in the proof of item 1. We give it below for
completeness.

For each $j\in I_0, r\in\{1,\ldots,m\},$ and $l\in\{1,\ldots,m\},$
let us define: \begin{align} \widetilde{p}_{jrl} =
\textmd{Pr}\left(P_{1j}\in\left(\frac{(l-1)\widetilde{q_1}}{m},
\frac{l\widetilde{q_1}}{m} \right], C_r^{(j)}\right).
\end{align}
 As in expression (\ref{sum1}), for each
$j\in I_0$ and $r\in\{1,\ldots,k\}$ one has:
\begin{align*}
\textmd{Pr}\left(P_{1j} \leq\frac{r\widetilde{q_1}}{m},\,
C_{r}^{(j)}\right)=\sum_{l=1}^r\widetilde{p}_{jrl}.
\end{align*}
Using this equality we obtain:
\begin{align}
\sum_{j\in I_0}\sum_{r=1}^k\frac{1}{r}\textmd{Pr}\left(P_{1j}
\leq\frac{r\widetilde{q_1}}{m},\, C_{r}^{(j)}\right)&=\sum_{j\in
I_0}\sum_{r=1}^k\sum_{l=1}^r\frac{1}{r}\widetilde{p}_{jrl}=\sum_{j\in
I_0}\sum_{l=1}^k\sum_{r=l}^{k}\frac{1}{r}\widetilde{p}_{jrl}\notag\\&\leq
\sum_{j\in
I_0}\sum_{l=1}^k\sum_{r=l}^{k}\frac{1}{l}\widetilde{p}_{jrl}
\leq\sum_{j\in
I_0}\sum_{l=1}^k\frac{1}{l}\sum_{r=1}^{k}\widetilde{p}_{jrl}.\label{last3}
\end{align}

Since $\cup_{r=1}^k C_r^{(j)}$ is
 a union of disjoint events, we obtain for each
$j\in I_0$ and $l\in\{1,\ldots,k\}$:
\begin{align}
\sum_{r=1}^k \widetilde{p}_{jrl}&=\textmd{Pr}\left(P_{1j}\in
\left(\frac{(l-1)\widetilde{q_1}}{m}, \frac{l\widetilde{q_1}}{m}
\right],\,
\cup_{r=1}^kC_r^{(j)}\right)\notag\\&\leq\textmd{Pr}\left(P_{1j}\in
\left(\frac{(l-1)\widetilde{q_1}}{m}, \frac{l\widetilde{q_1}}{m}
\right]\right)\notag\\&=
\textmd{Pr}\left(P_{1j}\leq\frac{l\widetilde{q_1}}{m}\right)-\textmd{Pr}\left(P_{1j}\leq\frac{(l-1)\widetilde{q_1}}{m}\right).\notag
\end{align}

Therefore for each $j\in I_0$ we obtain:
\begin{align}
\sum_{l=1}^k&\frac{1}{l}\sum_{r=1}^k\widetilde{p}_{jrl}\leq
\sum_{l=1}^k\frac{1}{l}\left[\textmd{Pr}\left(P_{1j}\leq
\frac{l\widetilde{q}_1}{m}\right)-\textmd{Pr}\left(P_{1j}\leq
\frac{(l-1)\widetilde{q}_1}{m}\right)\right]\notag\\&=
\sum_{l=1}^k\frac{1}{l}\textmd{Pr}\left(P_{1j}\leq
\frac{l\widetilde{q}_1}{m}\right)-\sum_{l=0}^{k-1}\frac{1}{l+1}\textmd{Pr}\left(P_{1j}\leq
\frac{l\widetilde{q}_1}{m}\right)\notag
\\&=\sum_{l=1}^{k-1}\left(\frac{1}{l}-\frac{1}{l+1}
\right)\textmd{Pr}\left(P_{1j}\leq\frac{l\widetilde{q}_1}{m}\right)+\frac{1}{k}\textmd{Pr}\left(P_{1j}\leq
\frac{k\widetilde{q}_1}{m}\right)\notag\\&\leq\sum_{l=1}^{k-1}\frac{1}{l+1}\left(\frac{\widetilde{q}_1}{m}\right)+\frac{\widetilde{q}_1}{m}=\left(\frac{\widetilde{q}_1}{m}\right)\sum_{l=1}^{k}\frac{1}{l}.\label{danilast3}
\end{align}
The inequality in (\ref{danilast3}) follows from the fact that for
$j\in I_0,$ $\textmd{Pr}(P_{1j}\leq x)\leq x$ for all $x\geq 0.$
Combining (\ref{danilast3}) with (\ref{last3}) we obtain an upper
bound for the first term of the sum in (\ref{sumitem3}):
\begin{align}\sum_{j\in I_0}\sum_{r=1}^k\frac{1}{r}\textmd{Pr}\left(P_{1j}
\leq\frac{r\widetilde{q_1}}{m},\, C_{r}^{(j)}\right)\leq\sum_{j\in
I_0}\left(\frac{\widetilde{q}_1}{m}\right)\sum_{l=1}^{k}\frac{1}{l}=\frac{|I_0|\,\widetilde{q}_1}{m}\sum_{l=1}^{k}\frac{1}{l}.\label{item13}\end{align}
We will now find an upper bound for the second term of the sum in
(\ref{sumitem3}):
\begin{align}
\sum_{j\in I_0}\sum_{r=k+1}^m\frac{1}{r}\textmd{Pr}\left(P_{1j} \leq
t,\, C_{r}^{(j)}\right)&=\sum_{j\in
I_0}\sum_{r=k+1}^m\frac{1}{r}\textmd{Pr}\left(P_{1j}\leq
t\right)\textmd{Pr}\left(C_r^{(j)}\,|\,P_{1j}\leq
t\right)\notag\\&\leq\sum_{j\in
I_0}\sum_{r=k+1}^m\frac{t}{r}\textmd{Pr}\left(C_r^{(j)}\,|\,P_{1j}\leq
t\right)\label{unif}\\&\leq\frac{\widetilde{q}_1}{m}\sum_{j\in
I_0}\sum_{r=k+1}^m\textmd{Pr}\left(C_r^{(j)}\,|\,P_{1j}\leq
t\right)\label{tr}\\&=\frac{\widetilde{q}_1}{m}\sum_{j\in
I_0}\textmd{Pr}\left(\cup_{r=k+1}^mC_r^{(j)}\,|\,P_{1j}\leq
t\right)\leq \frac{|I_0|\,\widetilde{q}_1}{m}.\label{lastlast}
\end{align}
The inequality in (\ref{unif}) follows from the fact that for $j\in
I_0,$ $\textmd{Pr}(P_{1j}\leq x)\leq x$ for all $x\geq 0.$ The
inequality in (\ref{tr}) follows from the fact that for all $r\geq
k+1,$ it holds that $r\geq \lceil
tm/\widetilde{q}_1-1\rceil+1=\lceil tm/\widetilde{q}_1 \rceil\geq
tm/\widetilde{q}_1,$ yielding that $t/r\leq \widetilde{q}_1/m.$ The
equality in (\ref{lastlast}) follows from the fact that
$\cup_{r=k+1}^mC_r^{(j)}$ is a union of disjoint events.

Combining (\ref{sumitem3}), (\ref{item13}) and (\ref{lastlast}) we
obtain an upper bound for the first term of the sum in (A.1):
\begin{align}
E\left(\frac{\sum_{j\in I_{0}}R_j}{\max(R, 1)}\right)\leq
\frac{|I_0|\,\widetilde{q}_1}{m}\sum_{l=1}^{k}\frac{1}{l}+
\frac{|I_0|\,\widetilde{q}_1}{m}\leq
\widetilde{q}_1\left(1+\sum_{l=1}^{k}\frac{1}{l}
\right)=q_1.\label{resitem3}\end{align} Note that for $t\leq q_1/m,$
$\lceil tm/q_1-1 \rceil=0,$ therefore
$q_1 = \max \{x:\,\,x(1+\sum_{i=1}^{\lceil mt/x-1 \rceil}1/i)=q_1 \}$. We obtain $\widetilde{q}_1=q_1,$ which yields
that in this case no modification is required.

It follows from Lemma \ref{lemdepSM}, item 3, that the second term
of the sum in (A.1) is bounded by $q_2.$ Combining this result with
(\ref{resitem3}), we obtain $$FDR\leq q_1+q_2=q_1+q-q_1=q.$$
\subsection{Proof of Lemma \ref{lemdepSM}}\label{subsec-prooflemmaSM2}
\textbf{Proof of item 1.}
Our proof is similar to the proof of Theorem 1.2 in \cite{yoav6}. For $j\in\{1,\ldots,m\}$ and $s\in\{1,\ldots,m-1\}$ we define the event $D_s^{(j)}$ 
as follows:
\begin{align*}
D_s^{(j)}=\{(P_1^{(j)}, P_2^{(j)}):\,T_{(s)}>s+1,
T_{(s+1)}>s+2,\ldots, T_{(m-1)}>m\},
\end{align*}
and we define $D_m^{(j)}$ to be the entire sample space of
$(P_1^{(j)}, P_2^{(j)}).$ Note that $D_s^{(j)}=\cup_{r=1}^s
C_r^{(j)}.$ It is easy to see that $D_s^{(j)}$ is the event in which
if $H_{NR,j}$ is rejected by Procedure 3.2, at most $s$ hypotheses
are rejected including $H_{NR,j}$.

We will first show that for each $p_1$, $j\in I_{10}\cap
\mathcal{R}_1(p_1)$ and $s\in\{1,\ldots,m-1\},$ $D_s^{(j)}\cap
\{P_1=p_1\}$ is an increasing set for $P_2^{(j)},$ i.e. if $(P_1,
P_2^{(j)})\in D_s^{(j)}\cap \{P_1=p_1\}$ and
$\widetilde{P}_2^{(j)}\geq
 P_2^{(j)},$ then $(P_1, \widetilde{P}_2^{(j)})\in D_s^{(j)}\cap
\{P_1=p_1\}.$ The result follows from the fact that for fixed
$P_1=p_1$ and $j\in I_{10}\cap \mathcal{R}_1(p_1),$ $T_i=\infty$ for
$i\notin\mathcal{R}_1^{(j)}(p_1^{(j)}),$ and  $T_i$ is increasing in
$P_{2i}$ for $i\in\mathcal{R}_1^{(j)}(p_1^{(j)}).$


 For a given $p_1$
and $j\in I_{10}\cap \mathcal{R}_1(p_1)$,
using the fact that for each 
$s\in\{1,\ldots,m-1\},$ 
$D_s^{(j)}\cap \{P_1=p_1\}$ is an increasing set for $P_2^{(j)},$
as well as the PRDS property of the $p$-values from the follow-up study and the
independence of the $p$-values across the studies, we obtain for
each  $s\in\{1,\ldots, R_1(p_1)-1\}$:
 \begin{align}
 \textmd{Pr}\left(D_s^{(j)}\,|\,P_{2j}\leq \frac{sq_2}{R_1(p_1)}, P_1=p_1\right)\leq \textmd{Pr}\left(D_s^{(j)}\,|\,P_{2j}\leq \frac{(s+1)q_2}{R_1(p_1)},
 P_1=p_1\right).\label{mainprds}
 \end{align}
Using the fact that for each $s\in\{1,\ldots, R_1(p_1)-1\},$
$D_s^{(j)}\cup C_{s+1}^{(j)}=D_{s+1}^{(j)},$ where $D_s^{(j)}$ and
$C_{s+1}^{(j)}$ are disjoint events, and the fact that
$D_1^{(j)}=C_1^{(j)}$ we obtain:
\begin{align}
&\sum_{r=1}^{R_1(p_1)} \textmd{Pr}\left(C_r^{(j)}\,|\,P_{2j}\leq
\frac{rq_2}{R_1(p_1)}, P_1=p_1\right)=\notag\\&
\textmd{Pr}\left(D_1^{(j)}\,|\,P_{2j}\leq \frac{q_2}{R_1(p_1)},
P_1=p_1\right)+\notag\\&\sum_{r=2}^{R_1(p_1)}\left[\textmd{Pr}\left(D_r^{(j)}\,|\,P_{2j}\leq
\frac{rq_2}{R_1(p_1)}, P_1=p_1
\right)-\textmd{Pr}\left(D_{r-1}^{(j)}\,|\,P_{2j}\leq
\frac{rq_2}{R_1(p_1)}, P_1=p_1 \right)
\right]\notag\\&=\sum_{r=1}^{R_1(p_1)}\textmd{Pr}\left(D_r^{(j)}\,|\,P_{2j}\leq
\frac{rq_2}{R_1(p_1)}, P_1=p_1
\right)-\sum_{r=1}^{R_1(p_1)-1}\textmd{Pr}\left(D_r^{(j)}\,|\,P_{2j}\leq
\frac{(r+1)q_2}{R_1(p_1)}, P_1=p_1
\right)\notag\\&\leq\sum_{r=1}^{R_1(p_1)}\textmd{Pr}\left(D_r^{(j)}\,|\,P_{2j}\leq
\frac{rq_2}{R_1(p_1)}, P_1=p_1
\right)-\sum_{r=1}^{R_1(p_1)-1}\textmd{Pr}\left(D_r^{(j)}\,|\,P_{2j}\leq
\frac{rq_2}{R_1(p_1)}, P_1=p_1
\right)\label{inprds}\\&=\textmd{Pr}\left(D_{R_1(p_1)}^{(j)}\,|\,P_{2j}\leq
q_2, P_1=p_1\right)=1,\notag
\end{align}
where the inequality in (\ref{inprds}) follows from
(\ref{mainprds}).

\textbf{Proof of item 2.} Let $p_1$ be
arbitrary fixed. Then,
\begin{align}
&E\left(\sum_{j\in I_{10}}R_j/\max(R,
1)\,|\,P_1=p_1\right)=\notag\\&\sum_{j\in I_{10}\cap
\mathcal{R}_1(p_1)}\sum_{r=1}^{R_1(p_1)}\frac{1}{r}\,\textbf{I}\left[p_{1j}\leq
\frac{rq_1}{m}\right]\textmd{Pr}\left(
P_{2j}\leq\frac{rq_2}{R_1(p_1)},
C_{r}^{(j)}\,|\,P_1=p_1\right)\notag
\\&\leq\sum_{j\in I_{10}\cap
\mathcal{R}_1(p_1)}\sum_{r=1}^{R_1(p_1)}\frac{1}{r}\textmd{Pr}\left(
P_{2j}\leq\frac{rq_2}{R_1(p_1)},
C_{r}^{(j)}\,|\,P_1=p_1\right)\label{fireq}\\&= \sum_{j\in
I_{10}\cap
\mathcal{R}_1(p_1)}\sum_{r=1}^{R_1(p_1)}\frac{1}{r}\textmd{Pr}\left(
P_{2j}\leq\frac{rq_2}{R_1(p_1)}\,|\,P_1=p_1\right)\textmd{Pr}
\left(C_{r}^{(j)}\,|\,P_{2j}\leq\frac{rq_2}{R_1(p_1)}, P_1=p_1
\right)\notag\\&\leq \frac{q_2}{R_1(p_1)}\sum_{j\in I_{10}\cap
\mathcal{R}_1(p_1)}\sum_{r=1}^{R_1(p_1)}\textmd{Pr}
\left(C_{r}^{(j)}\,|\,P_{2j}\leq\frac{rq_2}{R_1(p_1)}, P_1=p_1
\right)\leq\frac{q_2}{R_1(p_1)}|I_{10}\cap\mathcal{R}_1(p_1)|.\label{final}
\end{align}
The first inequality in (\ref{final}) follows from the independence
of the $p$-values across the studies and the fact that for each
$j\in I_{10},$ $\textmd{Pr}(P_{2j}\leq x)\leq x$ for all $x\geq 0.$
The second inequality in (\ref{final}) follows from Lemma
\ref{lemdepSM}, item 1. 
Taking the expectation over $P_1,$ we obtain $E\left(\sum_{j\in
I_{01}}R_j/\max(R,1)\right)\leq q_2.$

\textbf{Proof of item 3.} For $q_1'$ and $p_1$ arbitrary fixed,
\begin{align}
&E\left(\sum_{j\in I_{10}}R_j/\max(R, 1)\,|\,P_1=p_1\right)=\notag\\
&\sum_{j\in I_{10}\cap
\mathcal{R}_1(p_1)}\sum_{r=1}^{R_1(p_1)}\frac{1}{r}\,\textbf{I}\left[p_{1j}\leq
\frac{rq'_1}{m}\right]\textmd{Pr}\left(
P_{2j}\leq\frac{rq_2}{R_1(p_1)},
C_{r}^{(j)}\,|\,P_1=p_1\right)\notag
\\&\leq\sum_{j\in I_{10}\cap
\mathcal{R}_1(p_1)}\sum_{r=1}^{R_1(p_1)}\frac{1}{r}\textmd{Pr}\left(
P_{2j}\leq\frac{rq_2}{R_1(p_1)},
C_{r}^{(j)}\,|\,P_1=p_1\right).\notag
\end{align}
The arguments that lead from (\ref{fireq}) to the result
of item 2 complete the proof.

\section{Additional theoretical results under dependence}
\begin{theorem}\label{thmdepgen}
 Assume 
that the $p$-values across studies are independent, the $p$-values
within the primary study are independent, and the set of $p$-values
within the follow-up study has property PRDS. If the selection rule
used in step 1 of Procedure 3.2 is a valid selection rule, then
Procedure 3.2 with parameters $(q_1,q)$ controls the FDR at level
$q$ for the family of no replicability null hypotheses
$H_{NR,1},\ldots,H_{NR,m}$.
\end{theorem}
\begin{proof}
Let us first find an upper bound for the first term of the sum in
(A.1). Note that (A.3) is established using the independence of the
$p$-values within the primary study only, therefore it holds for any
form of dependence among the $p$-values within the follow-up study.
In particular, (A.3) holds under the dependency of Theorem
\ref{thmdepgen}, establishing an upper bound for the first term of
the sum in (A.1). It follows from Lemma \ref{lemdepSM}, item 2, that
the second term of the sum in (A.1) is bounded by $q_2.$ Thus we
obtain:
$$FDR\leq \frac{|I_0|q_1}{m}+q_2=\frac{|I_0|q_1}{m}+q-q_1\leq q.$$
\end{proof}
\begin{theorem}\label{SMgenthm}
 Assume that the $p$-values across studies are independent.  Then Procedure 3.2 with parameters $(q_1,q)$
 controls the FDR  at level
$q$ for the family of no replicability null hypotheses
$H_{NR,1},\ldots,H_{NR,m}$ if the selection
rule used in step 1 of Procedure 3.2 is a valid selection rule, and the expressions in step 2 of Procedure 3.2 are modified as follows:
\begin{enumerate}
\item  In the terms $r(q-q_1)/R_1$ and $R_2(q-q_1)/R_1,$ $q-q_1$ is replaced by $(q-q_1)/(\sum_{i=1}^{R_1} 1/i)$, and in the terms $rq_1/m$ and $R_2q_1/m,$ $q_1$ is replaced by $q_1/(\sum_{i=1}^m1/i)$.
\item  In the terms $r(q-q_1)/R_1$ and $R_2(q-q_1)/R_1,$ $q-q_1$ is replaced by $(q-q_1)/(\sum_{i=1}^{R_1} 1/i)$, and in the terms $rq_1/m$ and $R_2q_1/m,$ $q_1$ is replaced by
$\widetilde{q}_1,$ where
$$\widetilde{q}_1=\max\{x:\,\,x(1+\sum_{i=1}^{\lceil tm/x-1
\rceil}1/i)=q_1\},$$ if  only  hypotheses with primary study
$p$-values at most a fixed threshold $t<q_1/(1+\sum_{i=1}^{m-1}1/i)$
are considered for follow-up, i.e.
$\mathcal{R}_1\subseteq\{j\in\{1,\ldots,m\}:\,P_{1j}\leq t\}$.
\end{enumerate}
\end{theorem}
\textbf{Proof of item 1.} We will first show that the first term of
the sum in (A.1) is bounded by $|I_0|\,q_1/m.$ The first term of the
sum in (A.1) equals to:
\begin{align}
E\left(\frac{\sum_{j\in I_{0}}R_j}{\max(R, 1)}\right)&=\sum_{j\in
I_0}\sum_{r=1}^m\frac{1}{r}\textmd{Pr}\left(j\in \mathcal{R}_1,
P_{1j}\leq\frac{rq_1}{m\sum_{s=1}^m\frac{1}{s}},
P_{2j}\leq\frac{r(q-q_1)}{|\mathcal{R}_1|\sum_{s=1}^{|\mathcal{R}_1|}1/s},
C_{r}^{(j)}\right)\notag
\notag\\&\leq\sum_{j\in
I_0}\sum_{r=1}^m\frac{1}{r}\textmd{Pr}\left(P_{1j}
\leq\frac{rq_1}{m\sum_{s=1}^m\frac{1}{s}},
C_{r}^{(j)}\right).\notag
\end{align}
Now it follows from the arguments that lead from (\ref{foritem2}) to
(\ref{item1}) that the first term of the sum in (A.1) is bounded by
$|I_0|\,q_1/m.$


 Let us now find an upper bound for the second term of the sum
in (A.1). For each $p_1,$ $j\in \mathcal{R}_1(p_1)\cap I_{10},
r\in\{1,\ldots,R_1(p_1)\}$ and $l\in\{1,\ldots,R_1(p_1)\},$ let us
define:
\begin{align}p_{jrl}(p_1)=\textmd{Pr}\left(P_{2j}\in\left(\frac{(l-1)q_2}{R_1(p_1)\sum_{s=1}^{R_1(p_1)}1/s},
\frac{lq_2}{R_1(p_1)\sum_{s=1}^{R_1(p_1)}1/s} \right],\,
C_r^{(j)}\,\Big|\,P_1=p_1\right). \label{pjrlp1}\end{align} Note
that for each $p_1,$ $j\in \mathcal{R}_1(p_1)\cap I_{10}$ and $l\in
\{1,\ldots,R_1(p_1)\}$:
\begin{align}\sum_{r=1}^{R_1(p_1)}p_{jrl}(p_1)&=\textmd{Pr}\left(\cup_{r=1}^{R_1(p_1)}C_r^{(j)}, P_{2j}\in \left(\frac{(l-1)q_2}{R_1(p_1)\sum_{s=1}^{R_1(p_1)}1/s}, \frac{lq_2}{R_1(p_1)\sum_{s=1}^{R_1(p_1)}1/s}
\right]\,\Big|\,P_1=p_1\right)\notag\\&=\textmd{Pr}\left(P_{2j}\in
\left[\frac{(l-1)q_2}{R_1(p_1)\sum_{s=1}^{R_1(p_1)}1/s},
\frac{lq_2}{R_1(p_1)\sum_{s=1}^{R_1(p_1)}1/s}\right]
\,\Big|\,P_1=p_1\right).\label{whole2}
\end{align}
The equalities follow from the fact that given $P_1=p_1,$
$\cup_{r=1}^{R_1(p_1)}C_r^{(j)}$ is the whole sample space for
$P_2^{(j)}$, represented as a union of disjoint events. In addition,
note that for each $p_1,$ $j\in \mathcal{R}_1(p_1)\cap I_{10}$ and
$r\in\{1,\ldots,R_1(p_1)\},$
\begin{align}
\textmd{Pr}\left(P_{2j}
\leq\frac{rq_2}{R_1(p_1)\sum_{s=1}^{R_1(p_1)}1/s},
C_{r}^{(j)}\,|\,P_1=p_1\right)=\sum_{l=1}^{r}p_{jrl}(p_1),\label{sum332}
\end{align}
since for $j\in I_{10},$ $\textmd{Pr}(P_{2j}\leq x)\leq x$ for all
$x\geq 0,$ in particular $\textmd{Pr}(P_{2j}=0)=0.$ Therefore, for
each $p_1,$
\begin{align}
& E\left(\sum_{j\in I_{10}}R_j/\max(R, 1)\,|\,P_1=p_1\right) = \notag\\
&\sum_{j\in I_{10}\cap
\mathcal{R}_1(p_1)}\sum_{r=1}^{R_1(p_1)}\frac{1}{r}\,\textbf{I}\left[p_{1j}\leq
\frac{rq_1}{m\sum_{s=1}^m1/s}\right]\textmd{Pr}\left(
P_{2j}\leq\frac{rq_2}{R_1(p_1)\sum_{s=1}^{R_1(p_1)}1/s},
C_{r}^{(j)}\,|\,P_1=p_1\right)\notag\\&\leq\sum_{j\in I_{10}\cap
\mathcal{R}_1(p_1)}\sum_{r=1}^{R_1(p_1)}\sum_{l=1}^r\frac{1}{r}p_{jrl}(p_1)=\sum_{j\in
I_{10}\cap
\mathcal{R}_1(p_1)}\sum_{l=1}^{R_1(p_1)}\sum_{r=l}^{R_1(p_1)}\frac{1}{r}p_{jrl}(p_1)\label{for4}\\&\leq\sum_{j\in
I_{10}\cap
\mathcal{R}_1(p_1)}\sum_{l=1}^{R_1(p_1)}\sum_{r=l}^{R_1(p_1)}\frac{1}{l}p_{jrl}(p_1)\leq\sum_{j\in
I_{10}\cap
\mathcal{R}_1(p_1)}\sum_{l=1}^{R_1(p_1)}\frac{1}{l}\sum_{r=1}^{R_1(p_1)}p_{jrl}(p_1)\notag\\&=\sum_{j\in
I_{10}\cap
\mathcal{R}_1(p_1)}\sum_{l=1}^{R_1(p_1)}\frac{1}{l}\textmd{Pr}\left(P_{2j}\in
\left(\frac{(l-1)q_2}{R_1(p_1)\sum_{s=1}^{R_1(p_1)}1/s},
\frac{lq_2}{R_1(p_1)\sum_{s=1}^{R_1(p_1)}1/s}\right]
\,\Big|\,P_1=p_1\right)\notag\\&=\sum_{j\in I_{10}\cap
\mathcal{R}_1(p_1)}\sum_{l=1}^{R_1(p_1)}\frac{1}{l}\textmd{Pr}\left(P_{2j}\in
\left(\frac{(l-1)q_2}{R_1(p_1)\sum_{s=1}^{R_1(p_1)}1/s},
\frac{lq_2}{R_1(p_1)\sum_{s=1}^{R_1(p_1)}1/s}\right]
\right)\label{twogen},
\end{align}
where the first inequality in (\ref{for4}) follows from
(\ref{sum332}), the next to last equality follows from
(\ref{whole2}), and the equality in (\ref{twogen}) follows from the
independence of the $p$-values across the studies. Using similar
arguments to those leading to (\ref{danilast}), we obtain:
\begin{align*}\sum_{l=1}^{R_1(p_1)}\frac{1}{l}\textmd{Pr}\left(P_{2j}\in
\left(\frac{(l-1)q_2}{R_1(p_1)\sum_{s=1}^{R_1(p_1)}1/s},
\frac{lq_2}{R_1(p_1)\sum_{s=1}^{R_1(p_1)}1/s}\right]
\right)\leq\frac{q_2}{R_1(p_1)}.\end{align*} Combining this result
with (\ref{twogen}) we obtain for each
$p_1$:\begin{align*}E\left(\sum_{j\in I_{10}}R_j/\max(R,
1)\,|\,P_1=p_1\right)\leq\sum_{j\in I_{10}\cap
\mathcal{R}_1(p_1)}\frac{q_2}{R_1(p_1)}=\frac{|I_{10}\cap
\mathcal{R}_1(p_1)|}{R_1(p_1)}q_2\leq q_2.\end{align*} It follows
that \begin{align}E\left(\sum_{j\in I_{10}}R_j/\max(R, 1)\right)\leq
q_2.\label{res2}\end{align} Using this fact and the upper bound for
the first term of the sum in (A.1), we obtain that $FDR\leq
|I_0|q_1/m+q_2=|I_0|q_1/m+q-q_1\leq q.$

\textbf{Proof of item 2.} The first term of the sum in (A.1) equals
to:
\begin{align}
E\left(\frac{\sum_{j\in I_{0}}R_j}{\max(R, 1)}\right)&=\sum_{j\in
I_0}\sum_{r=1}^m\frac{1}{r}\textmd{Pr}\left(j\in \mathcal{R}_1,
P_{1j}\leq\frac{r\widetilde{q}_1}{m},
P_{2j}\leq\frac{r(q-q_1)}{|\mathcal{R}_1|\sum_{s=1}^{|\mathcal{R}_1|}1/s},
C_{r}^{(j)}\right)\notag
\notag\\&\leq\sum_{j\in
I_0}\sum_{r=1}^m\frac{1}{r}\textmd{Pr}\left(P_{1j}
\leq\min\left(\frac{r\widetilde{q}_1}{m},t\right),\,\,
C_{r}^{(j)}\right).\notag
\end{align}
 Now it follows from the arguments that lead from (\ref{min}) to
(\ref{resitem3}) that the upper bound for the first term of the sum
in (A.1) is $q_1.$

The second term of the sum in (A.1) is  $E\left(\sum_{j\in
I_{10}}R_j/\max(R, 1)\right).$ For each $p_1,$
\begin{align}
&E\left(\sum_{j\in I_{10}}R_j/\max(R,
1)\,|\,P_1=p_1\right) = \notag\\
&\sum_{j\in I_{10}\cap
\mathcal{R}_1(p_1)}\sum_{r=1}^{R_1(p_1)}\frac{1}{r}\,\textbf{I}\left[p_{1j}\leq
\frac{r\widetilde{q}_1}{m}\right]\textmd{Pr}\left(
P_{2j}\leq\frac{rq_2}{R_1(p_1)\sum_{s=1}^{R_1(p_1)}1/s},
C_{r}^{(j)}\,|\,P_1=p_1\right)\notag\\&\leq\sum_{j\in I_{10}\cap
\mathcal{R}_1(p_1)}\sum_{r=1}^{R_1(p_1)}\sum_{l=1}^r\frac{1}{r}p_{jrl}(p_1),\notag
\end{align}
where $p_{jrl}(p_1)$ is defined in (\ref{pjrlp1}). Now it follows
from the arguments that lead from (\ref{for4}) to (\ref{res2}) that
the second term of the sum in (A.1) is bounded by $q_2.$ Therefore,
$$FDR\leq q_1+q_2=q_1+q-q_1=q.$$

Consider now a situation where both studies are available before the
analysis, as described in Section 4 of the main manuscript. Without
loss of generality, we label the studies as study one and study two.
\begin{theorem}
Assume the $p$-values across studies are independent. Procedure 4.1
with parameters $(w_1, q_1, q)$ controls the FDR at level $q$ for
the family of no replicability null hypotheses
$H_{NR,1},\ldots,H_{NR,m}$ in either one of the following
situations:
\begin{enumerate}
\item The set of $p$-values within each study has property
PRDS, and the selection rule in step 1 of Procedure 3.2 is
Bonferroni at level $w_1q_1$ when the primary study is study one,
and at level $(1-w_1)q_1$ when the primary study is study two.
\item Arbitrary dependence among the $p$-values within each study, and
the expressions in step 2 of Procedure 3.2 are modified as follows:
in the terms $rq_1/m$ and $R_2q_1/R_1,$ $q_1$ is replaced by
$q_1/(\sum_{i=1}^m1/i),$ and in the terms $r(q-q_1)/R_1$ and
$R_2(q-q_1)/R_1,$ $q-q_1$ is replaced by $(q-q_1)/(\sum_{i=1}^{R_1}
1/i).$
\end{enumerate}
\end{theorem}
\begin{proof}
The proof of Theorem 4.1 in Appendix C relies only on the facts that
Procedure 3.2 used in step 1 and in step 2 of Procedure 4.1 is
valid. Therefore, the same proof shows that item 1 follows from
item 2 of Theorem 3.3,  and item 2 follows from item 1 of Theorem \ref{SMgenthm}.
\end{proof}


\end{document}